\documentclass[12pt]{article}
\usepackage{amsmath,amsfonts, amssymb}
\usepackage{natbib}
\usepackage{epsfig}
\usepackage[singlespacing]{setspace}
\usepackage{graphicx}
\usepackage{hyperref}
\usepackage{color}
\usepackage[bb=boondox]{mathalfa}
\usepackage[margin=1.05in]{geometry}
\usepackage{algorithm}
\usepackage[noend]{algpseudocode}
\usepackage[toc,page]{appendix}
\usepackage{blkarray}
\usepackage{caption}
\usepackage{subcaption}

\usepackage{tikz}
\tikzstyle{vertex}=[circle, draw, inner sep=1pt, minimum size=15pt]

\tikzstyle{block} = [draw, rectangle, 
    minimum height=3em, minimum width=6em]

\usetikzlibrary{arrows.meta}
\tikzset{>={Latex[width=2mm,length=2mm]}}

\newtheorem{definition}{\sc {Definition}}[section]
\newtheorem{theorem}{\sc {Theorem}}[section]
\newtheorem{lemma}{\sc {Lemma}}[section]
\newtheorem{corollary}{\sc {Corollary}}[section]

\newtheorem{example}{\sc {Example}}

\newfloat{procedure}{htbp}{loa}
\floatname{procedure}{Procedure}

\newenvironment{proof}[1][\bf{Proof.}]{\begin{trivlist}
\item[\hskip \labelsep {\bfseries #1}]}{\end{trivlist}}

\newenvironment{remark}[1][\bf{Remark.}]{\begin{trivlist}
\item[\hskip \labelsep {\bfseries #1}]}{\end{trivlist}}

\newcommand{\qed}{\hfill \mbox{\rule{2mm}{2mm}} \vspace{0.5cm}}

\begin{document}

\title{On Variants of Network Flow Stability}
\author{Young-San Lin \footnote{Computer Science Department, Purdue University, e-mail: {lin532@purdue.com}}  \and Th\`anh Nguyen \footnote{Krannert School of Management, Purdue University, e-mail: {nguye161@purdue.edu}}}
 \date{}
\maketitle
\onehalfspacing
\begin{abstract}
We consider a general stable flow problem in a directed and capacitated network, where each vertex has a strict preference list over the incoming and outgoing edges. A flow is stable if no group of vertices forming a path can mutually benefit by rerouting the flow. Motivated by applications in supply chain networks, we  generalize the traditional Kirchhoff's law, requiring the outflow is equal to the inflow at every nonterminal node,  to a monotone piecewise linear relationship between the inflows and the outflows. We show the existence of a stable flow using Scarf's Lemma, and provide a polynomial time algorithm to find such a stable flow. We further show that  finding a minimum cost  generalized stable network is NP-hard, while the problem is polynomial time solvable for 
the traditional stable flow satisfying Kirchhoff's law.

\end{abstract}


\section{Introduction}
In the classic stable marriage problem,  men and  women with individual preference orders of the opposite gender, are to be matched such that there are no man-woman pairs who are inclined to abandon their original partners and marry each other. \cite{gale1962college} showed the existence of such a stable matching by the deferred acceptance (DA) algorithm. Since then, the stable marriage problem and the DA algorithm have become  the cornerstones of market design and have changed the organization of many centralized markets, including resident matching, school choice and kidney exchange. (See for example, \cite{roth-peranson,roth2005kidney,Atila2003}.)

Following the success of the stable marriage problem, stability in supply chain networks, a generalization of the stable marriage problem from a two-sided market to a network market, has been studied extensively. See for example, \cite{Ostrovsky,fleiner2010stable,hatfield2015chain}. One natural version of the problem, first considered by \cite{fleiner2010stable}, is modeled as a directed graph whose vertices and edges represent agents and bilateral contracts, respectively. The outgoing endpoint vertex of an edge is the seller of that contract while the incoming endpoint vertex is the buyer. Each edge of the network has a capacity, representing the maximum amount of goods that can be traded in the contract. It is further assumed that each vertex holds a preference list over the adjacent contracts. A flow is a collection of bilateral contracts satisfying Kirchhoff's law, that is, the sum of inflow contracts is the same as the sum of outflow contracts for every vertex, except a source and a sink vertex representing a producer and a consumer, respectively. A flow is stable if no group of agents forming a path can benefit more by cooperating among themselves.  

Stable network flow has a variety of applications in physical and Internet traffic networks \cite{haxell2008fractional} as well as  supply chain and manufacturing networks \cite{fleiner2010stable}.  However,  because of the Kirchhoff's law, many applications do not fit this model.  For example, the amount of goods that an intermediate  firm receives is not equal to the amount the firm sells because during transportation the goods might get damaged and lost.  In the case of manufacturing firms, the relationship between input and output varies depending on the products and the manufacturing technologies.\footnote{See   \cite{flowbook} for other applications of general flow constraints.}

Motivated by this, we extend \cite{fleiner2010stable} by considering monotone and piecewise linear relationship between inflows and outflows at every node of the network. While the generalized flow problems have been studied extensively from the ``cardinal'' optimization point of view, to the best of our knowledge, there is essentially no prior work studying the problems with ordinal preferences.  

Our results are based on a series of reductions. As shown in \cite{fleiner2010stable}, the traditional stable flow problem can be reduced to the stable allocation problem. We show that the generalized stable flow problem can be eventually reduced to the generalized stable allocation (GSA) problem (\cite{dean2009generalized}). The whole reduction is non-trivial. We start from reducing this problem to a stable solution of Scarf's Lemma (\cite{scarf1967core}). However, finding a stable solution in general framework of Scarf is PPAD-hard (\cite{kintali}). Our next contribution is to provide another reduction from the Scarf's instance to GSA. We further study the optimization version of the problem and show that  finding a minimum cost  generalized stable network is NP-hard, while the problem is polynomial time solvable for  the traditional stable flow satisfying Kirchhoff's law.

Our contributions are thus twofold.   First, the generalized problem we introduce and our algorithm apply to  a wide range of applications in supply chains networks. Second, our hardness result shows that  generalizing Kirchhoff's law is not just a technical transformation between network flows, it significantly  changes the  problem's computational complexity.

The paper is organized  as follow. After discussing related work, in section~\ref{sec2}, we describe the setting and background definitions of this problem including how flow and stability are defined when agents utilize some special mappings to their inflow and outflow. Specifically, we introduce the concept of piecewise linear (PL) mapping, truncated linear (TL) mapping, and linear mapping.

In section~\ref{sec3}, we show the existence of stable flow in a truncated linear network (TL-network) where inflow and outflow has a TL relation for all agents by a reduction to Scarf's Lemma (see \cite{scarf1967core}). Later on, by reducing a piecewise linear network (PL-network) where inflow and outflow has a PL relation for all agents to another TL-network, we can show the existence of stable flow in a PL-network.


In section~\ref{sec4}, we show the entire framework from reducing TL-network flow stability to GSA. The main approach is to first map the TL-network instance to a Scarf's instance mentioned in section~\ref{sec3}, then reduce the Scarf's instance to a linear network with three layers, eventually reduce the new network to GSA. As a PL-network instance can be reduced to a TL-network instance, the entire problem for finding stable flow in a PL-network is enclosed. 


%
In section~\ref{sec5}, we consider an optimization variant of the stable flow problem where edges have associated cost. For the traditional problem, finding stable flow with minimum cost is polynomial time solvable. However, finding the minimum cost  a stable flow in AL-networks, TL-networks, and PL-networks are all NP-hard.

\subsubsection*{Related Work}
As discussed above, our paper is closely related to \cite{fleiner2010stable}, which provides a reduction of the stable flow to the stable allocation problem considered in \cite{baiou2002stable}. The reduction, therefore,  shows the existence of a stable flow. \cite{cseh2013stable} provide  a preflow-push variant of the Gale-Shapley algorithm to find a stable flow  in pseudo-polynomial time. Shortly afterward,  \cite{cseh2013new} used a path augmenting variant of the Gale-Shapley algorithm to improve the running time to be polynomial. To the best of our knowledge, there have not been existence results nor polynomial time algorithms to find a stable solution for the generalized stable flow problem considered in this paper.


There is an extensive optimization literature on generalized flow problem, that assume the total outflow is equal to a constant times the total inflow. (See for example, \cite{kantorovich1960mathematical,flowbook} and recent work of \cite{vegh}). This literature however mainly focused on the ``cardinal'' problem, that is, to find the maximum total  flow. Our paper provides a first step to investigate ``ordinal'' preferences in the generalized flow problem.



Our paper is also related to a growing economic literature on network stability problem. \cite{ostrovsky2008stability} is the first to introduce the concept of stability to supply chain networks, and followed by, among others, \cite{hatfield2013stability, hatfield2015chain} and \cite{fleiner2016trading}. However,  these papers focus on discrete choice function and integral stable solutions. In particular, in these models each edge in the network corresponds to a trade and an outcome is a set of trades. In our paper each trade has a continuous capacity, and the  difficulties  are the non convex constraints between the inflow and outflow for each vertex.  
 
Recently, \cite{che2015stable,azevedo2015existence} consider continuous model of matching in {\em two-sided markets}, and show the existence of a stable solution. However, these results are based on fixed-point theorem argument and polynomial time algorithms to find such a stable solution are thus not known.

\section{Preliminaries} \label{sec2}


\subsection{Stable Flow}
We start with traditional stable flow introduced in \cite{fleiner2010stable}.  
A network is a quadruple $(G,s,t,c)$, where $G=(V\cup\{s,t\},E)$ is a simple digraph. $s$ and $t$ are the source and sink vertices and $c: E \to \mathbb{R}_+$ determines the capacity $c(e)$ for  $e \in E$. $s$ only has outgoing edges while $t$ only has incoming edges. For each $v \in V$, there is a path from $s$ to $v$ and a path from $v$ to $t$. The strict preference $\succ_v$ of a vertex $v \in V$ is defined over edges adjacent to $v$. Incoming edges and outgoing edges are ranked strictly and separately. $e_1 \succ_v e_2$ means $v$ prefers $e_1$ to $e_2$.

A {\em flow} of network $(G,s,t,c)$ is a function $f: E \to \mathbb{R}_{\geqslant 0}$ such that capcity condition $0 \leqslant f(e) \leqslant c(e)$ holds for each $e \in E$ and each vertex $v \in V$ satisfies Kirchhoff's Law: $\sum_{uv \in E}f(uv) = \sum_{vu \in E}f(vu)$, that is, the amount of inflow equals the amount of outflow.

Let $f$ be a flow, a walk $W=(v_1, v_2, ..., v_k)$ is {\em blocking} $f$ if all the followings hold:
\begin{enumerate}
\item $v_i v_{i+1}$ is unsaturated for $i=1, ..., k-1$.
\item $v_1 = s$ or there is an edge $v_1 u \in E$ such that $v_1 v_2 \succ_{v_1} v_1 u$ and $f(v_1 u) > 0$.
\item $v_k = t$ or there is an edge $u v_k \in E$ such that $v_{k-1} v_k \succ_{v_k} u v_k$ and $f(u v_k) > 0$.
\end{enumerate}

An edge $e$ is {\em unsaturated} if $f(e) < c(e)$. A flow $f$ is {\em stable} if there is no blocking walk. Besides, while defining flow stability, we also assume that for any vertex in $V$, routing flow on any of the adjacent edges are preferred over no routing.

Network flow stability is a general model capturing the classic stable marriage problem as a sub case.  To see this,  consider the classic stable marriage problem with a set $\mathcal{M}$ of $n$ men and a set $\mathcal{W}$ of $n$ women with individual strict preference of the opposite gender. We also assume that each man or woman rather have a partner than be unmatched. In a matching $M$, we denote $M(m)$ as the partner of $m$ for $m \in \mathcal{M}$ and $M(w)$ as the partner of $w$ for $w \in \mathcal{W}$. A pair $(m,w)$ is blocking $M$ if all the followings hold:

\begin{enumerate}
\item $m$ and $w$ are not matched in $M$.
\item $m$ is unmatched in $M$ or $m$ prefers $w$ to $M(m)$.
\item $w$ is unmatched in $M$ or $w$ prefers $m$ to $M(w)$.
\end{enumerate}

We can construct a network by attaching $s$ to men vertices, men vertices to women vertices, and women vertices to $t$, by unit capacity edges. For preferences, each $m \in \mathcal{M}$ prefers $w_1$ to $w_2$ if and only if $mw_1 \succ_{m} mw_2$ and each $w \in \mathcal{W}$ prefers $m_1$ to $m_2$ if and only if $m_1w \succ_{w} m_2w$. Herein a blocking pair corresponds to a blocking walk in the network. Example~\ref{ex1} shows how network flow stability can be applied to the stable marriage problem.

\begin{example} \label{ex1}
\mbox{}\\
\begin{center}
\begin{tikzpicture}
	\node (M) at (0,5.5) {$\mathcal{M}$};
	\node (W) at (4,5.5) {$\mathcal{W}$};
	\node[vertex] (s) at (-2,3) {$s$};
	\node[vertex] (t) at (6,3) {$t$};
	\node[vertex] (m1) at (0,4.5) {$m_1$};
	\node[vertex] (m2) at (0,3) {$m_2$};
	\node[vertex] (m3) at (0,1.5) {$m_3$};
	\node[vertex] (w1) at (4,4.5) {$w_1$};
	\node[vertex] (w2) at (4,3) {$w_2$};
	\node[vertex] (w3) at (4,1.5) {$w_3$};
	\path[->]
		(s) edge [dashed] (m1)
		(s) edge (m2)
		(s) edge (m3)
		(m1) edge [dashed] (w1)
		(m2) edge [dashed] node [above, pos=0.2] {$1^{\mbox{st}}$} (w1)
		(m2) edge node [pos=0.2] {$3^{\mbox{rd}}$} (w2)
		(m2) edge [dashed] node [below, pos=0.2] {$2^{\mbox{nd}}$} node [above, pos=0.8] {$1^{\mbox{st}}$} (w3)
		(m3) edge node [below, pos=0.8] {$2^{\mbox{nd}}$} (w3)
		(w1) edge [dashed] (t)
		(w2) edge (t)
		(w3) edge (t);
\end{tikzpicture}
\end{center}
This network corresponds to a stable marriage matching instance. $m_2 w_1 \succ_{m_2} m_2 w_3 \succ_{m_2} m_2 w_2$ and $m_2 w_3 \succ_{w_3} m_3 w_3$. Although $\mathcal{M}$ and $\mathcal{W}$ form a complete bitartite graph, we highlight the edges with flow value 1 by solid edges while other edges have flow value 0. The dashed edges form blocking walks. The flow corresponds to the matching $\{(m_2,w_2),(m_3,w_3)\}$. The blocking pair $(m_1, w_1)$ corresponds to the blocking walk $(s, m_1, w_1, t)$; the blocking pair $(m_2, w_1)$ corresponds to the blocking walk $(m_2, w_1, t)$; while the blocking pair $(m_2, w_3)$ corresponds to the blocking walk $(m_2, w_3)$.
\end{example}

Network flow stability not only can be applied to stability in two sided markets, but also to stability in multi-level markets or supply chain networks. It is convenient to depict these settings as a digraph, where each capacitated edge $e \in E$ represents a contract with limited product quantity, and each vertex $v \in V$ represents an intermediate agent that holds its individual preference of incoming and outgoing contracts and strives to maximize the amount of flow through $v$. A blocking walk of a flow interprets a scenario in which a group of agents are willing to cooperate selfishly and benifit from rerouting some flow. A flow is stable if every agent is satisfied to its current offer such that no group of agents have the incentive to reroute the flow since they cannot benifit more.

\subsection{Generalized Stable Flow}
To define a general network flow, we use the same notation as the traditional stable flow problem, except we consider some special classes of mappings for vertices that we use to model the constraint between inflow and outflow. In our paper, a mapping is more general than a function because it maps an element of $\mathbb{R}_{\geqslant 0}$ to an element or an interval of $\mathbb{R}_{\geqslant 0}$. We are interested in linear, truncated linear and piecewise linear mappings. They are defined as follows.  


\begin{definition} \label{def2.1}
$g$ is a {\bf linear mapping} if  $g(x) = a x$  for  $a>0$. $g$ is a {\bf truncated linear mapping} if it is one of the following: 
\begin{equation} \label{eq1} 
g(x) = 
\begin{cases} 
      [0,b] & \text{if } x=0, \\
      a x + b & \text{otherwise.} \\
   \end{cases}
\end{equation}
		
\begin{equation} \label{eq2} g(x) = 
\begin{cases} 
      0 & \text{if } x<\frac{b}{a}, \\
      a x - b & \text{otherwise.} \\
   \end{cases}
			\end{equation}
where $a > 0$ and $b \geqslant 0$.
\end{definition}

The following figure illustrates the two possible cases of a truncated linear mapping.
		\begin{center}
    \begin{tikzpicture}
      \draw[->] (-0.2,0) -- (2,0) node[right] {$x$};
      \draw[->] (0,-0.2) -- (0,2) node[above] {$g(x)$}; 
      \node at (-0.2,0.5) {$b$};
      \node at (1.3,1) {$a x + b$};
\draw[scale=0.5,domain=0:1,smooth,variable=\x,blue] plot (0,{\x}) [scale=0.5,domain=0:6,smooth,variable=\x,blue] plot({\x},{\x+2});

      \draw[->] (4.8,0) -- (7,0) node[right] {$x$};
      \draw[->] (5,-0.2) -- (5,2) node[above] {$g(x)$};
 	  \node at (5.5,0.4) {$\frac{b}{a}$};
      \node at (7.6,1) {$a x - b$};
 \draw[scale=0.5,domain=10:11,smooth,variable=\x,blue] plot ({\x},0) [scale=0.5,domain=11:16,smooth,variable=\x,blue] plot({\x+11},{\x-11});
    \end{tikzpicture}
		\end{center}

\begin{definition} \label{def2.2}
A {\bf piecewise linear mapping} concatenates $k$ segments of truncated linear mappings, namely:
    \begin{equation} \label{eq3} g(x) = 
\begin{cases} 
      g_{1}(x) & \text{if } x \in [c_0, c_1] \\
      \max\{g_{1}(c_1)\} + g_{2}(x-c_1) & \text{if } x \in [c_1, c_2], \\
			\max\{g_{1}(c_1)\} + \max\{g_{2}(c_2 - c_1)\} + g_{3}(x - c_2) & \text{if } x \in [c_2, c_3], \\
      \vdots \\
      \sum_{i=1}^{k-1}{\max\{g_{i}(c_i - c_{i-1})\}} + g_{k}(x-c_{k-1}) & \text{if } x \in [c_{k-1}, \infty).
   \end{cases}
		\end{equation}
where $c_0 = 0$, $c_0 \leqslant c_1 \leqslant ... \leqslant c_{k-1}$, and $g_{i}$ is a truncated linear  mapping for $i=1, 2, ..., k$. 
\end{definition}


The following figure is an example of piecewise linear mapping.

\begin{center}
   \begin{tikzpicture}
      \draw[->] (-0.7,0) -- (5,0) node[right] {$x$};
      \draw[->] (-0.5,-0.2) -- (-0.5,5) node[above] {$g(x)$};
			\node at (-0.2,-0.2) {$c_0$};
      \node at (0.5,-0.2) {$c_1$};
      \node at (1.5,-0.2) {$c_2$};
      \node at (3,-0.2) {$c_3 = c_4$};
      \node at (0,4.7) {$g_{1}$};
      \node at (1,4.7) {$g_{2}$};
      \node at (2.25,4.7) {$g_{3}$};
      \node at (3,5) {$g_{4}$};
      \node at (4,4.7) {$g_{5}$};
      \draw[domain=0:0.5,smooth,variable=\x,blue] plot (-0.5,{\x}) [domain=-0.5:0.5,smooth,variable=\x,blue] plot({\x},{0.25*(\x+0.5)+0.5});
    \draw[domain=0.5:1.5,smooth,variable=\x,blue] plot({\x},{1.2*(\x-0.5)+0.75});
    \draw[domain=1.5:2.2,smooth,variable=\x,blue] plot({\x},{1.95})
[domain=2.2:3,smooth,variable=\x,blue] plot({\x},{0.7*(\x-2.2)+1.95});
    \draw[domain=2.51:3,smooth,variable=\x,blue] plot(3,{\x})
[domain=3:4,smooth,variable=\x,blue] plot({\x},{3})
[domain=4:5,smooth,variable=\x,blue] plot({\x},{\x-1});
	  \draw[domain=0:5,dotted,variable=\x,red] plot (0.5,{\x});
      \draw[domain=0:5,dotted,variable=\x,red] plot (1.5,{\x});
			\draw[domain=0:5,dotted,variable=\x,red] plot (3,{\x});
    \end{tikzpicture}
\end{center}

Note that $g_2(x)$ is a linear mapping while $g_4(x)$ maps to only an interval of $\mathbb{R}_{\geqslant 0}$.


In a network $(G,s,t,c)$, for each $v \in V$, $v$ is associated with a mapping $g_v$ where $x$ is the inflow of $v$ and $g_v(x)$ is the outflow of $v$. It is convenient to introduce the following definition assuming that each vertex in a network all apply some specific mappings.

\begin{definition} \label{def2.3}
If for all $v \in V$, $g_v$ is a truncated linear mapping, then the network is a {\bf truncated linear network} (TL-network). If for all $v \in V$, $g_v$ is a piecewise linear mapping, then the network is a {\bf piecewise linear network} (PL-network).
\end{definition}

In a market network, each vertex can be regarded as an agent given offers of incoming and outgoing contracts. They evaluate the quantity of desired outgoing contracts to be signed based on how many incoming contracts are accepted. Therefore, the feasibility of contract assignment can be defined as the following:

\begin{definition} \label{def2.4}
Given a flow $f$ of a TL-network or PL-network, for each $v \in V$, let $f_{in}(v) = \sum_{uv \in E}f(uv)$ and $f_{out}(v) = \sum_{vu \in E}f(vu)$, $f$ is {\bf feasible} if:
\begin{enumerate}
\item $0 \leqslant f(e) \leqslant c(e)$ for each $e \in E$.
\item For each $v \in V$, $g_v(f_{in}(v)) = f_{out}(v)$ if $g_v(f_{in}(v))$ maps $f_{in}(v)$ to an element of $\mathbb{R}_{\geqslant 0}$; $g_v(f_{in}(v)) \in f_{out}(v)$ if $g_v(f_{in}(v))$ maps $f_{in}(v)$ to an interval of $\mathbb{R}_{\geqslant 0}$.
\end{enumerate}
\end{definition}

The definition of flow stability is defined in the similar way to the traditional definition:

\begin{definition} \label{def2.5}
Given a flow $f$ of a TL-network or PL-network, $f$ is stable if it is feasible and there is no blocking walk in the network. $f$ has a blocking walk $W = (v_1, v_2, ..., v_{k-1}, v_k)$ where $v_i v_{i+1} \in E$ for $i=1, ..., k-1$ if all the followings hold:

\begin{enumerate}
\item There exists a vector $V_W = (r_1, r_2, ..., r_{k-1})$ such that:

\begin{enumerate}
\item  $r_i \geqslant 0$ and there is at least one $r_i > 0$ for $i=1, ..., k-1$.
\item $r_i \leqslant c(v_iv_{i+1}) - f(v_iv_{i+1})$ for $i=1, ..., k-1$.
\item For $i=2, ..., k-1$, $f_{out}(v_i)+r_i = g_{v_i}(f_{in}(v_i)+r_{i-1})$ if $g_{v_i}$ maps $f_{in}(v_i)+r_{i-1}$ to an element of $\mathbb{R}_{\geqslant 0}$; $f_{out}(v_i)+r_i \in g_{v_i}(f_{in}(v_i)+r_{i-1})$ if $g(f_{in}(v))$ maps to an interval of $\mathbb{R}_{\geqslant 0}$.
\end{enumerate}
\item $v_1 = s$ or there is an edge $v_1 u \in E$ such that $v_1 v_2 \succ_{v_1} v_1 u$ and $f(v_1 u) > 0$.
\item $v_k = t$ or there is an edge $u v_k \in E$ such that $v_{k-1} v_k \succ_{v_k} u v_k$ and $f(u v_k) > 0$.
\end{enumerate}
\end{definition}

All the conditions in Definition~\ref{def2.5} are the same as the ones for the traditional stable flow, except the first condition has to be modified because the flow no longer satisfies Kirchhoff's Law. Point 1.(a) shows that in the blocking walk $W$, at least one agent $v_i$ has the incentive to deliver positive flow. Point 1.(b) restricts the flow value that is intended to be delivered along $W$ by the remaining capacity, while point 1.(c) checks the feasibility. In other words, $W$ is blocking $f$ if vertices in $W$ are better off by rerouting a feasible flow bounded by the remaining capacity. $f$ is stable if no group of vertices can benifit from rerouting the flow. For an example of stable and unstable flow in a TL-network, see Example~\ref{ex2} in the appendix.

We denote {\em TL-stable-flow} (TL-SF) a stable flow assignment in a TL-network, or the problem of finding a stable flow in a TL-network, and {\em PL-stable-flow} (PL-SF) a stable flow assignment in a PL-network, or the problem of finding stable flow in PL-networks.


\section{Flow Stability of TL-networks and PL-networks} \label{sec3}


Scarf's Lemma originally appeared as a tool to prove the non-emptiness of the core in an $n$ person game (\cite{scarf1967core}). In this section, we show the existence of TL-SF by a reduction to Scarf's Lemma. The existence of PL-SF, on the other hand, is shown by a reduction to the existence of TL-SF.

\subsection{Scarf's Lemma}

\begin{definition} \label{def3.1} Let $A$ be an $m \times n$ nonnegative matrix with at least one positive entry in every column and row, $b \in \mathbb{R}^m_+$ be a positive vector, and $\mathcal{P} = \{x: x \geqslant 0, Ax \leqslant b\}$. For each row $i$ of $A$, there is a strict ranking $\succ_i$ over the columns in $\{1 \leqslant j \leqslant n: A_{ij} > 0\}$. $k \succ_i j$ means row $i$ prefers column $k$ to column $j$.

We say $x \in \mathcal{P}$ {\bf dominates} column $j$ if there exists a row $i$ such that:
\begin{enumerate}
\item $A_{ij} > 0$ and the constraint $i$ binds, i.e. $(Ax)_i = b_i$.
\item $k \succ_i j$ for any other $k \neq j$ such that $A_{ik} > 0$ and $x_k > 0$.
\end{enumerate}
\end{definition}

To simplify our notation, given an $x \in \mathcal{P}$, we also say row $i$ {\it dominates} column $j$ if the above mentioned conditions hold.

\begin{lemma}[Scarf's Lemma] \label{lem3.1} For any above mentioned $A$, $b$, and $\succ_i$, there exists an $x^* \in \mathcal{P}$ that dominates all columns of $A$.
\end{lemma}

\subsection{From TL-SF to Scarf's Lemma} \label{subsubsec3.2.1}

Given a TL-network $(G,s,t,c)$, where $G=(V\cup\{s,t\},E)$, and $\succ_v$ for any $v \in V$, in order to show the existence of TL-SF, we employ Scarf's Lemma, i.e. construct the corresponding matrix $A$ and vector $b$. By finding a vector that dominates all the columns of $A$, we obtain a corresponding TL-SF.

\begin{theorem} \label{thm3.1}
There exists a TL-SF in a TL-network.
\end{theorem}

See Appendix~\ref{app:thm3.1} for the formal proof. Here we present the high level proof strategy and the construction of $A$ and $b$. $A$ and $b$ should capture the flow capacity and feasibility constraints along with a proper setting for the row preferences. For the vector $x$, it will be convenient to introduce the variable $x_e$ that stands for the flow value of $e \in E$ and the corresponding column index of $A$.

We start with the edge capacity constraints:

\begin{enumerate}
\item For each $e \in E$, create column $x_e$ for $A$ and row $e$ for $A$ and $b$.
\item Set $A_{e x_e}=1$, other elements of $A$ in row $e$ are zeros, and $b_{e} = c(e)$.
\end{enumerate}

This allows us to assure that the flow of edge $e$ will not exceed its capacity $c(e)$. Notice that $A_{e x_e}$ is the only positive element in row $e$, so row $e$ only prefers column $x_e$. The part for capturing the capacity constraints involves an $|E| \times |E|$ sub-matrix of $A$ and an $|E|$ dimensional sub-vector of $b$.

To ensure flow feasibility and stability, we create two rows $v^{in}$ and $v^{out}$ and introduce two more auxiliary variables $x^{in}_v$ and $x^{out}_v$ for each vertex $v \in V$. Row $v^{in}$ handles the inflow of $v$ and row $v^{out}$ handles the outflow of $v$. Besides, $v$ is associated with a truncated linear mapping $g_v$, we add $v$ as a subscript of the parameters in Definition~\ref{def2.1}. That is, $g_v$ has parameters $a_v$ and $b_v$ involved. The following is the remaining construction of $A$:

\begin{enumerate}
\item For each $v \in V$, create row $v^{in}$ and $v^{out}$.
\item For each $v \in V$, create column $x^{in}_v$ and column $x^{out}_v$.
\item For $e=uv \in E$, set $A_{v^{in} x_e} = a_v$.
\item For $e=vu \in E$, set $A_{v^{out} x_e} = 1$.
\item For each $v \in V$, $A_{v^{in} x^{in}_v} = A_{v^{in} x^{out}_v}=A_{v^{out} x^{in}_v} = A_{v^{out} x^{out}_v}=1$.
\item Row $v^{in}$ prefers column $x^{in}_v$ the most and column $x^{out}_v$ the least. Preference of $e = uv \in E$ remains the same as $\succ_v$.
\item Row $v^{out}$ prefers column $x^{out}_v$ the most and column $x^{in}_v$ the least. Preference of $e = vu \in E$ remains the same as $\succ_v$.
\end{enumerate}

Notice that all the entries not mentioned are set to zeros. The remaining is to select proper values for $b_{x^{in}_v}$ and $b_{x^{out}_v}$. To achieve this, let $M(v) = \max(\sum_{uv \in E} c(uv), \sum_{vu \in E} c(vu), b_v) + 1$ be a large value such that at least one of $x^{in}_v$ and $x^{out}_v$ has to be positive. The remaining construction of $b$ is:

\begin{enumerate}
\item If $g_v$ is in the form of equation~\ref{eq1}, set $b_{v^{in}} = M(v)$ and $b_{v^{out}} = M(v) + b_v$.
\item If $g_v$ is in the form of equation~\ref{eq2}, set $b_{v^{in}} = M(v) + b_v$ and $b_{v^{out}} = M(v)$.
\end{enumerate}

The entire setting makes $A$ an $(|E| + 2|V|) \times (|E| + 2|V|)$ matrix and $b$ an $|E| + 2|V|$ dimensional vector. For vertex $v$, the difference between $b_{v^{out}}$ and $b_{v^{in}}$ is set to $b_v$. Depending on whether the inflow or outflow part is truncated, $b_{v^{out}}$ and $b_{v^{in}}$ are set to proper values accordingly.

By the setting of the preferences of row $v^{in}$ and $v^{out}$, variables $x^{in}_v$ and $x^{out}_v$ enable us to capture the stability and feasibility of a flow. When the inflow or outflow of $v$ does not reach the truncated threshold, the inflow or outflow is zero and one of the rows $v^{in}$ and $v^{out}$ dominates both columns $x^{in}_v$ and $x^{out}_v$. One of the rows $v^{in}$ and $v^{out}$ binds and its corresponding flow value is zero, while the other does not bind since the truncated threshold is not reached. When the inflow or outflow of $v$ reaches the truncated threshold, then row $v^{in}$ dominates column $x^{out}_v$ and row $v^{out}$ dominates column $x^{in}_v$. Flow feasibility is guaranteed because both row $v^{in}$ and $v^{out}$ bind.

Let $x^*$ be the Scarf's solution that dominates all the columns of $A$. With a bit abuse of notation, we label columns of $A$ by superscript or subscript of $x$. The superscript or subscript of $x^*$ stands for the exact value of the Scarf's solution $x^*$. For each $e \in E$, we can set $f(e) = {x^*}_e$. $x^*$ corresponds to a TL-SF by the setting of $A$ and $b$. On the other hand, a TL-SF corresponds to a vector $x^*$ that dominates all the columns of $A$. For more details of the proof, see Appendix~\ref{app:thm3.1}. For an example of a reduction from TL-SF to Scarf's Lemma, see Example~\ref{ex3} in the appendix.

\subsection{From PL-SF to TL-SF} \label{subsubsec3.2.2}

Given a PL-network $(G,s,t,c)$, where $G=(V\cup\{s,t\},E)$, and $\succ_v$ for any $v \in V$, as the existence of TL-SF is shown in Theorem~\ref{thm3.1}, we reduce PL-SF to TL-SF to show the existence of PL-SF.

\begin{corollary} \label{cor3.1}
There exists a PL-SF in a PL-network.
\end{corollary}

For each $v \in V$, $v$ is associated with a piecewise linear mapping $g_v$, we add $v$ as a subscript of the parameters in Definition~\ref{def2.2}. That is, $g_v$ has $k_v$ segments of truncated linear mappings $g_{v,i}$ and $c_{v,i-1}$ as segment boundaries where $i=1, ..., k_v$.

We create a sub-TL-network for each $v \in V$ as the following:
\begin{enumerate}
\item Split $v$ into $v_{in}$ and $v_{out}$ where $g_{v_{in}}(x)=x$ and $g_{v_{out}}(x)=x$.
\item For each $uv \in E$ (for simplicity, suppose $s_{out}=s$ and $t_{in} = t$), connect $u_{out}$ to $v_{in}$ by keeping the same capacity and preference as $uv$, that is, capacity of $u_{out} v_{in}$ is the same as $uv$, and $v_{in}$ has the same preference for incoming edges as $v$.
\item For each segment that applies $g_{v,i}$, create vertex $v_i$ where $g_{v_i}(x) = g_{v,i}$ and edges $v_{in} v_i$ and $v_i v_{out}$.
\item For capacities of $v_{in} v_i$ and $v_i v_{out}$:
\begin{enumerate}
\item Set $c(v_{in} v_i)=c_{v,i} - c_{v,i-1}$ and $c(v_i v_{out})=\max\{g_{v,i}(c_{v,i} - c_{v,i-1})\}$.
\item By (a) or (b), if $c(v_{in} v_i)=0$ then reset $c(v_{in} v_i) = \epsilon$; if $c(v_i v_{out})=0$ then reset $c(v_i v_{out}) = \epsilon$, for some small $\epsilon > 0$.
\end{enumerate}
\item $v_{in}$ prefers $v_{in} v_i$ and $v_{out}$ prefers $v_{out} v_i$ with smaller $i$.
\end{enumerate}

In the new TL-network, we split $v$ into $v_{in}$ and $v_{out}$ to handle incoming and outgoing edges in the same behavior as $v$ in the original PL-network, then create gadget vertices $v_i$ to handle the truncated linear mapping segments and prioritize the segments with smaller $i$ accordingly. Point 4.(c) deals with extreme cases to ensure that each edge has a positive capacity in the PL-network. For an illustration of the construction, see Example~\ref{ex4} in the Appendix.
\section{Reductions of Stable Flow Problems} \label{sec4}

In this section, we show the entire framework of reducing PL-SF to the generalized stable allocation (GSA) problem \cite{dean2009generalized}. We denote a Scarf's instance that corresponds to a TL-SF instance {\em TL-Scarf}. The intermediate step is to convert a TL-Scarf instance to a linear network with three layers. We denote the problem of finding stable flow in such networks {\em 3-layer-linear-SF}. Eventually, the GSA instance can be simply constructed from the 3-layer-linear-SF instance.


The following flowchart is an overview of the reductions between different problems.

\begin{center}
\begin{tikzpicture}
	\node[block] (mplm) at (0,0) {PL-SF};
	\node[block] (lm) at (6,0) {TL-SF};
	\node[block] (s) at (12,0) {TL-Scarf};
	\node[block] (aln) at (10,-3) {3-layer-linear-SF};
	\node[block] (gsa) at (2,-3) {GSA};
	\path[->]
		(mplm) edge node [above] {Section~\ref{subsubsec3.2.2}} (lm)
		(lm) edge node [left] {Corollary~\ref{cor5.1}} (aln)
		(lm) edge node [above] {Section~\ref{subsubsec3.2.1}} (s)
		(s) edge node [right] {Theorem~\ref{thm5.1}} (aln)
		(lm) edge node [left] {Corollary~\ref{cor5.3}} (gsa)
		(aln) edge [bend right=5] node [above] {Corollary~\ref{cor5.3}} (gsa)
		(gsa) edge [bend right=5] node [below] {Corollary~\ref{cor5.2}} (aln);
\end{tikzpicture}
\end{center}

\subsection{From TL-SF to 3-layer-linear-SF}

This reduction consists of two steps, from TL-SF to TL-Scarf shown in section~\ref{subsubsec3.2.1} and from TL-Scarf to 3-layer-linear-SF. Herein, we focus on the later part.

We will use the same notation as in section~\ref{subsubsec3.2.1} for the TL-Scarf instance with matrix $A$ and vector $b$. Beside the TL-Scarf instance, we are also given the corresponding TL-network $(G,s,t,c)$. The strategy is to create a bipartite-like graph with three layers. Recall that for $v \in V$ in the original TL-network, the rows $v^{in}$ and $v^{out}$ are created in the TL-Scarf instance. The first layer consists of vertices $v_{in}$ that correspond to rows $v^{in}$ while the third layer consists of vertices $v_{out}$ that correspond to rows $v^{out}$. Vertices in the second layer, on the other hand, either connects $u_{out}$ and $v_{in}$ for $uv \in E$ with a proper multiplier or connects $v_{out}$ and $v_{in}$ to prevent parallel edges appearing.

To construct the three-layer linear network $(G',s',t',c')$ where $G'=(V' \cup \{s', t'\}, E')$, for the vertex part $V'$:

\begin{enumerate}
\item Create $s'$, $s_{out}$, $t'$, and $t_{in}$.
\item For each row $v^{out}$, create vertex $v_{out}$.
\item For each row $v^{in}$, create vertex $v_{in}$.
\item For each column $x_e$, create vertex $m_e$.
\item For each column $x^{out}_v$, create vertex $m^{out}_v$.
\item For each column $x^{in}_v$, create vertex $m^{in}_v$.
\end{enumerate}

Starting from source $s'$, the first layer consists of the outgoing vertices $s_{out}$ and $v_{out}$, the second layer consists of the middle vertices $m_e$, $m^{out}_v$, and $m^{in}_v$, and the third layer consists of the incoming vertices $t_{in}$ and $v_{in}$, which evetually merge together and end at sink $t'$.

For the edge part $E'$ and capacities $c'$, we would like to caputure the TL-Scarf constraints. For $sv \in E$, the column $x^{sv}$ can only dominated by either the capacity constraint row $e$ or constraint row $v^{in}$. There is no row constraints associated with $s$, therefore, we should set $c'(s's_{out})$ in the new network large enough such that $s's_{out}$ will never be saturated. Similar argument applies for $c'(t_{in} t')$. For $v \in V$, $c'(s'v_{out})$ and $c'(v_{in}t')$ are set in a way to capture the upper bound of the TL-Scarf constraints. Capacities of edges connected to $m_e$ are set in a way to accommodate the flow through $m_e$ after applying its mutiplier. Finally, we connect $v_{in}$ and $v_{out}$ to $m^{out}_v$ and $m^{in}_v$ with proper capacities in order to prevent parallel edges and capture the TL-Scarf constraints. The setting of the edges and capacities is:

\begin{enumerate}
\item Set $c'(s' s_{out}) = \sum_{su \in E} {b_{su}} + 1$, and set $c'(s' v_{out}) = b_{v^{out}}$.
\item Set $c'(t_{in} t') = \sum_{ut \in E} {b_{ut}} + 1$, and set $c'(v_{in} t') = b_{v^{in}}$.
\item For each $e = uv \in E$, if $v \neq t$, set $c'(u_{out} m_e) = b_e$ and $c'(m_e v_{in}) = A_{v^{in} x_e} b_e$; if $v=t$, set $c'(u_{out} m_e) = c'(m_e t_{in}) = b_e$.
\item For each $v \in V$, set $c'(v_{out} m^{in}_v) = c'(v_{out} m^{out}_v) = c'(m^{in}_v v_{in}) = c'(m^{out}_v v_{in}) = \max\{b_{v^{in}}, b_{v^{out}}\}$.
\end{enumerate}

Point 1 sets the edge capacities between $s'$ and the first layer vertices while point 2 sets the edge capacities between the third layer vertices and $t'$. Point 3 and point 4 set the edge capacities from the first to the second layer and from the second to the third layer. The capacities are set in a way such that any vertex $v_{out}$ in the first layer must send some positive flow to either $v_{out} m^{in}_v$ or $v_{out} m^{out}_v$. Similarly, any vertex $v_{in}$ in the third layer must receive some positive flow from either $m^{in}_v v_{in}$ or $m^{out}_v v_{in}$. To match the multipliers in the TL-Scarf instance, the setting of the linear mapping for each vertex is:

\begin{enumerate}
\item The outflow of $s_{out}$, $t_{in}$, $v_{in}$, $v_{out}$, $m^1_v$, $m^2_v$, and $m_e$ where $e = vt$ are the same as their inflow.
\item The outflow of $m_e$ where $e=vw$ and $w \neq t$ is $A_{w^1 x_e}$ times of its inflow.
\end{enumerate}

Eventually, base on the setting of the TL-Scarf instance, we set the preference of vertices as the following:

\begin{enumerate}
\item For each $v \in V$, $v_{out}$ prefers $v_{out} m^{out}_v$ the most and $v_{out} m^{in}_v$ the least, the preference of $v_{out} m_e$ for some $e = vu$ is the same as in row $v^{out}$.
\item For each $v \in V$, $v_{in}$ prefers $m^{in}_v v_{in}$ the most and $m^{out}_v v_{in}$ the least, the preference of $m_e v_{in}$ for some $e = uv$ is the same as in row $v^{in}$.
\item The preference of $s_{out}$ and $t_{in}$ is arbitrary.
\end{enumerate}

Suppose we have a stable flow $f'$ of this three-layer linear network, we show that by setting ${x^*}_e = f'(u_{out} m_e)$ for $e=uv$, ${x^*}^{in}_v = f'(v_{out} m^{in}_v)$, and ${x^*}^{out}_v = f'(v_{out} m^{out}_v)$, we have a solution for TL-Scarf. On the other hand, the solution of TL-Scarf also corresponds to a stable flow. See Appendix~\ref{app:thm5.1} for the proof details of Theorem~\ref{thm5.1} and Example~\ref{ex9} in the Appendix for an example.

\begin{theorem} \label{thm5.1}
$f'$ is a stable flow of the above mentioned three-layer linear network if and only if $x^*$ is a solution of TL-Scarf.
\end{theorem}

A TL-Scarf instance is actually defined by a hidden TL-network. Therefore, we can shorten the two-stage reduction from TL-SF to TL-Scarf and from TL-Scarf to 3-layer-linear-SF into a reduction from TL-SF to 3-layer-linear-SF. See Appendix~\ref{app:cor5.1} for the proof.

\begin{corollary} \label{cor5.1}
For a TL-network, there is an equivalent three-layer linear network. That is, a stable flow of the three-layer linear network corresponds to a stable flow of the original TL-network.
\end{corollary}

\subsection{Stable Flow and Generalized Stable Allocation} \label{app:subsec5.1}

In this section, we introduce the generalized stable allocation (GSA) problem and show a reduction from TL-SF to GSA and a reduction from GSA to 3-layer-linear-SF.

\subsubsection{The Generalized Stable Allocation Problem}

The original stable allocation problem was stated in \cite{baiou2002stable}. Here we consider a more general case, the {\em generalized stable allocation} (GSA) problem. With a bit abuse of notation, GSA may stand for the problem itself or an allocation that is stable. The problem setting includes two finite disjoint sets $I$ and $J$ of jobs and machines respectively and an edge set $C$ (parallel edges are allowed) that forms a bipartite graph $G = (I \cup J, C)$. The function $u: C \to \mathbb{R}_+$ maps an egde to its capacity. Each job $i \in I$ has its own preference $\succ_i$ on edges in $C$ that have $i$ as an endpoint, and each machine $j \in J$ has its own preference $\succ_j$ on edges in $C$ that have $j$ as an endpoint. Each edge has a multiplier $\mu(ij) > 0$ that tells us one unit of job $i$ uses up $\mu(ij)$ units of capacity when assigned to machine $j$.

\begin{definition} \label{def5.1}
Let $x(ij)$ denote the quantity that job $i$ is assigned to machine $j$, an {\bf assignment} $x$ is {\bf feasible} if all of the followings hold:

\begin{align}
& \sum_{j \in J}{x(ij)} \leqslant 1 & \forall i \in I \label{eq7}\\
& \sum_{i \in I}{\mu(ij) x(ij)} \leqslant 1 & \forall j \in J  \label{eq8}\\
& 0 \leqslant x(ij) \leqslant u(ij) & \forall ij \in C \label{eq9}
\end{align}
\end{definition}

If the equality holds in equation~\ref{eq7}, then job $i$ is {\em $x$-saturated}. Similarly, if the equality holds in equation~\ref{eq8}, then machine $j$ is $x$-saturated. Edge $ij$ is saturated if $x(ij) = u(ij)$ in equation~\ref{eq9}. Now we can define stability:

\begin{definition} \label{def5.2}
An assignment $x$ is {\bf stable} if it is feasible and for each $ij \in C$, at least one of the followings hold:

\begin{enumerate}
\item $ij$ is saturated.
\item $i$ is $x$-saturated and $ij' \succ_i ij$ holds for any other $j' \neq j$ such that $x(ij') > 0$.
\item $j$ is $x$-saturated and $i'j \succ_j ij$ holds for any other $i' \neq i$ such that $x(i'j) > 0$.
\end{enumerate}
\end{definition}

In short, an assignment is stable if there is no blocking edge. A blocking edge is an unsaturated pair where the involved job and machine both mutually prefer this edge to some of their other currently assigned edges.

\begin{remark}
Our setting is different from the one in \cite{dean2009generalized} but they are equivalent. In their setting, adding a dummy job and a dummy machine always guarantees that all jobs and all machines except the dummy machine are $x$-saturated in a feasible assignment.
\end{remark}

\subsubsection{The Reductions}

The bi-directional reducibility between the stable allocation and the traditional stable flow problem was shown in \cite{fleiner2010stable}. We overview analogous statements, a reduction from GSA to 3-layer-AL-SF and a reduction from TL-SF to GSA.

\begin{corollary} \label{cor5.2}
GSA can be reduced to 3-layer-linear-SF.
\end{corollary}

\begin{proof}
For an instance of GSA including $G = (I \cup J, C)$, $\succ_i$ for $i \in I$, $\succ_j$ for $j \in J$, and the edge capacity function $u$, we construct the instance of 3-layer-linear-SF including $(G,s,t,c)$, $G=(V \cup \{s,t\},E)$, $\succ_{v}$ for $v \in V$, and $g_{v}$ for $v \in V$:

\begin{enumerate}
\item Create $s$, $t$, $V = \{v_i | i \in I\} \cup \{v_j | j \in J\} \cup \{v_{ij} | ij \in C\}$.
\item Set $c(s v_i) = 1$ for every $i \in I$ and $c(v_j t) = 1$ for every $j \in J$.
\item Set $c(v_i v_{ij}) = u(ij)$, $c(v_{ij} j) = \mu(ij) u(ij)$, and $g_{v_{ij}}(x) = \mu(ij) x$ for every $ij \in C$.
\item $v_i v_{ij} \succ_{v_i} v_i v_{ij'}$ in the 3-layer-linear-SF instance if and only if $ij \succ_i ij'$ in the GSA instance.
\item $v_{ij} v_j \succ_{v_j} v_{i'j} v_j$ in the 3-layer-linear-SF instance if and only if $ij \succ_j i'j$ in the GSA instance.
\end{enumerate}

By setting $x(ij) = f(v_i v_{ij})$, we have a solution for GSA. One can check that we have GSA if and only if we have 3-layer-linear-SF. If the assignment is not stable, then there exists a walk such that all three conditions in Definition~\ref{def2.5} do not hold, as a result, we can find a blocking walk from $s$ or some $v_i$ to $t$ or some $v_j$. If we can find a blocking walk in our 3-layer-linear-SF instance, then from the structure of the three-layer linear network, it must be from $s$ or some $v_i$ to $t$ or some $v_j$. In either case, all three conditions in Definition~\ref{def5.2} do not hold.
\qed
\end{proof}

\begin{corollary} \label{cor5.3}
TL-SF can be reduced to GSA.
\end{corollary}

\begin{proof}
As shown in Corollary~\ref{cor5.1}, TL-SF can be reduced to 3-layer-linear-SF. Therefore, we can just follow the analogous approach described in Corollary~\ref{cor5.1} and use the same notations. Suppose we apply Corollary~\ref{cor5.1} on the orignal TL-network and obtain the three-layer linear network. That is, we are given the three-layer linear network $(G',s',t',c')$ where $G'=(V'\cup\{s',t'\},E')$. Construct the GSA instance with $G = (I \cup J, C)$ as the following:

\begin{enumerate}
\item Let $I = \{i_s\} \cup \{i_v | v_{out} \in V'\}$ and $J = \{j_t\} \cup \{j_v | v_{in} \in V'\}$.
\item Let $C = \{i_w j_v | m_{wv} \in E'\} \cup \{i_w j_w | m^{out}_v \in V'\} \cup \{j_w i_w | m^{in}_v \in V'\}$.
\item Set $u(i_w j_v) = \frac{c'(w_{out} v_{in})}{c'(s' w_{out})}$ for each $m_{wv} \in E'$.
\item Set $u(i_w j_w) = u(j_w i_w) = \frac{c'(w_{out} m^{out}_w)}{c'(s' w_{out})}$ for each $m^{out}_w \in E'$.
\item Set $\mu(i_w j_v) = \frac{a_v c'(s' w_{out})}{c'(v_{in} t')}$ and $\mu(i_w j_w) = \mu(j_w i_w) = \frac{c'(s' w_{out})}{c'(w_{in} t')}$.
\item For each $w_{out} \in V'$, $i_w$ prefers $i_w j_w$ the most and $j_w i_w$ the least, the preference of $i_w$ for $i_w j_v$ where $m_{wv} \in E'$ is the same as $\succ_{w_{out}}$.
\item For each $v_{in} \in V'$, $j_v$ prefers $j_v i_v$ the most and $i_v j_v$ the least, the preference of $j_v$ for $i_w j_v$ where $m_{wv} \in E'$ is the same as $\succ_{v_{in}}$.
\item The preference of $i_s$ and $j_t$ is arbitrary.
\end{enumerate}

Clearly, the flow is stable in the original TL-network if and only if the allocation is stable in the new GSA instance by analogous argument in Corollary~\ref{cor5.1}. By setting $f(wv) = c'(s' w_{out}) x(i_w j_v)$ for each $wv \in E$, we have a solution of TL-SF.
\qed
\end{proof}

\begin{remark}
Although the bipartite graph $G$ as an undirected graph, $i_v j_v$ and $j_v i_v$ will be regarded as different edges for convenience.
\end{remark}

The GSA construction is actually equivalent to the three-layer lienar network in Corollary~\ref{cor5.1} with some proper capacity scaling. For an example, see Example~\ref{ex10} in the Appendix.

\subsection{Analysis}
We analyze the time complexity for finding both PL-SF and TL-SF.

%
%

\begin{corollary} \label{cor5.4}
PL-SF and TL-SF can be found in polynomial time.
\end{corollary}

\begin{proof}
Let us start with TL-SF. We can simply apply Corollary~\ref{cor5.3} and get an equivalent GSA then apply the algorithm described in \cite{dean2009generalized}. Suppose we are given a TL-network $(G,s,t,c)$ where $G=(V\cup\{s,t\},E)$, the equivalent GSA has $O(|V|)$ machines and jobs, and $O(|E|)$ edges. For the algorithm in \cite{dean2009generalized}, the length of the path to augment in each iteration is $O(|V|)$. The number of augmentations is $O(|E|)$. It takes $O(|E||V|)$ time.

Suppose we are given a PL-network with graph $(G,s,t,c)$ where $G=(V\cup\{s,t\},E)$, first reduce it to a TL-network then reduce TL-SF to GSA. The final equivalent GSA will have $O(|E|+K)$ vertices and $O(|E|+K)$ edges, where $K = \sum_{v \in V} {k_v}$. Recall that $k_v$ is the number of segments for each vertex's piecewise linear mapping, so $K$ is the total number of segments. By the special sturcture of the new TL-network and GSA, the length of path to augment in each iteration is $O(|V|)$ since edges from different segments of the same vertex cannot be augmented at the same time by the algorithm in \cite{dean2009generalized}, and the number of augmentations is $O(|E|+K)$. Therefore, it takes $O((|E|+K)|V|)$ time. Note that the information of each linear piecewise segment is a part of the input, so $K$ is polynomial in the size of the input.
\qed
\end{proof}

%

\begin{remark}
Using dynamic tree implementation,  we can actually  design a faster algorithm in \cite{dean2009generalized}. As a path is represented by a tree structure while using dynamic trees, all of the above mentioned can be done in $O(\log |V|)$ time where the length of the path to augment is $O(|V|)$. Therefore, this also yields $O(|E| \log |V|)$ bound for finding TL-SF and $O((|E|+K) \log |V|)$ bound for finding PL-SF.

\begin{corollary} \label{cor5.5}
TL-SF can be found in $O(|E| \log |V|)$ time and PL-SF can be found in $O((|E|+K) \log |V|)$ time.
\end{corollary}
\end{remark}

\section{The Optimal Stable Flow Problem} \label{sec5}

Given a network $(G,s,t,c)$, suppose while assigning flow values, we care about not only stability but also the cost. For each edge $e \in E$, $p: E \to \mathbb{R}_{\geqslant 0}$ is the {\em price per flow}. Given a feasible flow assignment $f$, the cost of this edge is $f(e) p(e)$. The total cost of the flow assignment is $\sum_{e \in E}{f(e) p(e)}$. The decision version of the stable flow optimization problem is defined as the following:

\begin{definition}[The Optimal Stable Flow Problem] \label{def7.1}
Given a network $(G,s,t,c)$, the price per flow function $p$, and the budget $B \in \mathbb{R}_{\geqslant 0}$, determine whether there exists a stable flow such that the total cost is at most $B$.
\end{definition}

We show that the traditional optimal stable flow problem is polynomial time solvable, while the optimal TL-SF and PL-SF are all NP-complete.

\subsection{The Optimal Stable Flow Problem}
For this problem, we assume that inflow is equal to outflow for each vertex. Before proving optimal stable flow problem is polynomial time solvable, we introduce the stable allocation problem. This is because the proof is based on the reduction from the stable flow problem to the stable allocation problem shown in \cite{fleiner2010stable}.

The stable allocation problem is the same as GSA except that each job or machine $k \in I \cup J$ has a quota $q(k)$ and the multiplier $\mu(ij)$ of each edge is one, that is, one unit of job $i$ uses up only one unit of capacity when assigned to machine $j$.

%
%
%
%

\begin{theorem} \label{thm6.1}
The optimal stable flow problem is polynomial time solvable.
\end{theorem}

\begin{proof}
We reduce this problem to the optimal stable allocation problem. Suppose we are given $(G,s,t,c)$, $G=(V\cup\{s,t\},E)$, $\succ_v$ for each $v \in V$, and the cost function $p$, construct the optimal stable allocation instance as the following:

\begin{enumerate}
\item Let $I = \{i_s\} \cup \{i_v | v \in V\}$ and $J = \{j_t\} \cup \{j_v | v \in V\}$.
\item Let $C = \{i_w j_v | wv \in E\} \cup \{i_v j_v | v \in V\} \cup \{j_v i_v | v \in V\}$.
\item Set $u(i_w j_v) = c(wv)$.
\item Let $M(v) = \max\{\sum_{wv \in E}{c(wv)}, \sum_{vw \in E}{c(vw)}\} + 1$ for each $v \in V \cup \{s,t\}$.
\item Set $q(i_v) = M(v)$ for each $i_v \in I$ and $q(j_v) = M(v)$ for each $j_v \in J$.
\item Set $u(i_v j_v) = u(j_v i_v) = \max\{\sum_{wv \in E}{c(wv)}, \sum_{vw \in E}{c(vw)}\} + 1$ for each $v \in V \cup \{s,t\}$.
\item For each $w \in V$, $i_w$ prefers $i_w j_w$ the most and $j_w i_w$ the least, the preference of $i_w$ for $i_w j_v$ where $wv \in V$ is the same as $\succ_w$.
\item For each $v \in V$, $j_v$ prefers $j_v i_v$ the most and $i_v j_v$ the least, the preference of $j_v$ for $i_w j_v$ where $wv \in E$ is the same as $\succ_v$.
\item Let $h(i_w j_v) = c(wv)$ for each $wv \in E$. The costs of the rest of the edges in $C$ are 0.
\end{enumerate}

By solving the optimal stable allocation instance and setting $f(wv) = x(i_w j_v)$, we obtain the stable flow with the minimum cost. See the details in \cite{fleiner2010stable} for to proof of stability. Clearly, the costs of the stable flow and the stable allocation are the same. The optimal stable allocation problem is polynomial time solvable according to \cite{dean2010faster}, so the reduction indicates that the optimal stable flow problem can also be solved in polynomial time.
\qed
\end{proof}

\subsection{Optimality of Generalized Stable Flow}

To show that the optimal PL-SF and TL-SF are NP-complete, we start with a problem that is NP-hard, the {\em optimal generalized stable allocation} (optimal GSA) problem.

The optimal GSA problem is the same as GSA except each edge $ij \in C$ has a cost $h(ij)$ where $h: C \to \mathbb{R}_{\geqslant 0}$. The total cost of an assignment $x$ is $\sum_{ij \in C}{h(ij) x(ij)}$. We consider not only stability but also whether the total cost is less than a given budget $H$ where $H \in \mathbb{R}_{\geqslant 0}$.
\cite{dean2009generalized}) show a reduction of the optimal GSA problem to the subset sum problem to obtain the following result.

\begin{theorem}[\cite{dean2009generalized})] \label{thm6.0}
Optimal GSA problem is NP-complete.
\end{theorem}

We show that PL-SF and TL-SF are all NP-complete by first presenting a polynomial time certificate and second reducing optimal GSA to these problems.

\begin{theorem} \label{thm6.2}
Optimal PL-SF, TL-SF, and AL-SF are all NP-complete.
\end{theorem}

\begin{proof}
Given a flow assignment, it is straightforward to check if a flow is feasible and the budget is exceeded or not in polynomial time.

For stability, we cannot check if each walk of a network is a blocking walk or not since there may be a lot of walks in a network. For a PL-network, we apply the reduction in section~\ref{subsubsec3.2.1} to get an equivalent TL-network, while for a TL-network, we apply the reduction in Corollary~\ref{cor5.3} to obtain an equivalent GSA. Checking if an assignment is stable only takes polynomial time, so checking if the correpsonding flow assignment is stable for the original PL-network or TL-network only takes polynomial time.

Therefore, we can check if a flow is feasible and stable, and the cost is below budget in polynomial time. Optimal PL-SF and TL-SF are all in NP.

For NP-hardness, as 3-layer-linear-SF is a special case of TL-SF and PL-SF, it suffices to reduce optimal GSA to an optimal stable flow problem in a three-layer linear network. We apply the reduction in Corollary~\ref{cor5.2} and additionally set $p(v_i v_{ij}) = h(ij)$ for each $ij \in C$ while the cost of other edges are set to zero. Clearly, the costs of the stable flow and the stable assignment are the same. Therefore, we have a stable assignment below budget if and only if we have a stable flow below budget.
\qed
\end{proof}
\section{Conclusion}
Our paper introduced a  general version of stable flow problem. We show the existence of a stable flow and a polynomial time algorithm to find such a solution.  This provides an initial step in the analysis of generalized flows with ordinal preferences.  We beleive that our results has significant implications in the design and analysis of supply chain networks with agents' hetergenous preferences.

We further show that unlike the traditional stable flow proble, the minimum cost stable generalized flow is NP hard.  This shows that the generalized stable flow is fundamentally more difficult than its original formulation.


%
%
%

\subsubsection*{Acknowledgment}
This research is partly supported by  National Science Foundation Grants AST-
1443965, CMMI 1728165.

\bibliographystyle{abbrvnat}
\bibliography{reference}
\appendix

\begin{center}
{\bf \huge Appendix}
\end{center}

\section{Proof of Theorem~\ref{thm3.1}}\label{app:thm3.1}
We prove the correctness of the reduction in section~\ref{subsubsec3.2.1} by showing that $f$ is a TL-SF if and only if $x^*$ is a Scarf's solution.

Scarf $\to$ TL-SF:

Suppose $x^*$ is a solution of Scarf, then $x^*$ dominates every column of $A$. First we have to show flow feasibility of $f$.

Assume row $v^{in}$ dominates column $x^{in}_v$ and $x^{out}_v$, then row $v^{in}$ prefers other nonzero entries (column $x_e$ where $e = uv \in E$ and column $x^{out}_v$) to column $x^{in}_v$. Row $v^{in}$ prefers column $x^{in}_v$ the most, so all other entries of $x^*$ must be zero and ${x^*}^{in}_v = b_{v^{in}}$. In order to obey constraint row $v^{out}$, $g_v$ must follow equation~\ref{eq1}, since otherwise $b_{v^{out}} < b_{v^{in}} = {x^*}^{in}_v$ implies that constraint row $v^{out}$ will be violated. Therefore, ${x^*}^{in}_v = M(v)$, ${x^*}^{out}_v = 0$, and $f_{in}(v) = \sum_{e = uv \in E}{x^*}_e = 0$. From constraint row $v^{out}$:

\[\sum_{e = vu \in E}{x^*}_e + {x^*}^{in}_v + {x^*}^{out}_v \leqslant M(v) + b_v \implies f_{out}(v) = \sum_{e = vu \in E}{x^*}_e \leqslant b_v\]

By similar argument, if row $v^{out}$ dominates column $x^{in}_v$ and $x^{out}_v$, $g_v$ must follow equation~\ref{eq2}, so ${x^*}^{in}_v = 0$, ${x^*}^{out}_v = M(v)$, and $f_{out}(v)  = \sum_{e = vu \in E}{x^*}_e = 0$. From constraint row $v^{in}$:

\[a_v \sum_{e = uv \in E}{x^*}_e + {x^*}^{in}_v + {x^*}^{out}_v \leqslant M(v) + b_v \implies f_{in}(v) = \sum_{e = vu \in E}{x^*}_e \leqslant \frac{b_v}{a_v}\]

The remaining case is row $v^{in}$ dominates column $x^{out}_v$ and row $v^{out}$ dominates column $x^{in}_v$. Both row $v^{in}$ and row $v^{out}$ bind, so we have either:

\[\begin{cases} 
      a_v \sum_{e = uv \in E}{x^*}_e + {x^*}^{in}_v + {x^*}^{out}_v = M(v) \\
      \sum_{e = vu \in E}{x^*}_e + {x^*}^{in}_v + {x^*}^{out}_v = M(v) + b_v
   \end{cases} \implies
a_v \sum_{e = uv \in E}{x^*}_e + b_v = \sum_{e = vu \in E}{x^*}_e\]

if $g_v$ follows equation~\ref{eq1} or

\[\begin{cases} 
      a_v \sum_{e = uv \in E}{x^*}_e + {x^*}^{in}_v + {x^*}^{out}_v = M(v) + b_v \\
      \sum_{e = vu \in E}{x^*}_e + {x^*}^{in}_v + {x^*}^{out}_v = M(v)
   \end{cases} \implies
a_v \sum_{e = uv \in E}{x^*}_e - b_v = \sum_{e = vu \in E}{x^*}_e\]

if $g_v$ follows equation~\ref{eq2}.

In either case, $f_{out}(v) = g_v(f_{in}(v))$. Besides, if $g_v$ follows equation~\ref{eq1}, then $\sum_{e = vu \in E}{x^*}_e \geqslant b_v$. Similarly, if $g_v$ follows equation~\ref{eq2}, then $\sum_{e = uv \in E}{x^*}_e \geqslant \frac{b_v}{a_v}$.

For stability, given $x^*$ dominates all columns of $A$, we show a proof by contradiction. Suppose $f$ is not stable by assigning $f(e) = {x^*}_e$ for each $e \in E$. That is, there is a blocking walk $W=(v_1, ..., v_k)$ with vector $V_W = (r_1, ..., r_{k-1})$ that satisfies the three conditions in Definition~\ref{def2.5}. Let $r_i$ be the first positive entry in $V_W$, consider the following cases:

\begin{enumerate}
\item If $i>1$, then $r_{i-1}=0$. By point 1.(c) in Definition~\ref{def2.5}, the only way to let $f_{out}(v_i)$ increase and $f_{in}(v_i)$ remain the same is $g_{v_i}$ follows equation~\ref{eq1} and $f(v_{i-1} v_i)=0$. Therefore, row ${v_i}^{in}$ dominates column $x^{in}_{v_i}$ and $x^{out}_{v_i}$, ${x^*}^{in}_{v_i}=M(v_i)$, and ${x^*}^{out}_{v_i}=0$. This indicates row ${v_i}^{out}$ does not dominate column $x_{v_i v_{i+1}}$ because $x^{in}_{v_i}>0$ and $x^{in}_{v_i}$ is the least preferred column. Row $v_i v_{i+1}$ does not dominate column $x_{v_i v_{i+1}}$ because edge $v_i v_{i+1}$ still has at least $r_i$ remaining capacity. Hence, row ${v_{i+1}}^{in}$ must dominate column $x_{v_i v_{i+1}}$.

\item If $i=1$, then row $v_1 v_2$ does not dominate column $x_{v_1 v_2}$ because edge $v_1 v_2$ still has at least $r_1$ remaining capacity. By point 2 in Definition~\ref{def2.5}, if $v_1=s$ then row ${v_2}^{in}$ must dominate column $x_{v_1 v_2}$ because there are no corresponding rows in $A$ for vertex $s$; if $v_1 v_2 \succ_{v_1} v_1 u$ for some other edge $v_1 u \in E$ and $f(v_1 u) > 0$, then row ${v_1}^{out}$ does not dominate column $x_{v_1 v_2}$ because it prefers column $x_{v_1 u}$ more, so row ${v_2}^{in}$ must dominate column $x_{v_1 v_2}$.
\end{enumerate}

In either case, row ${v_{i+1}}^{in}$ must dominate column $x_{v_i v_{i+1}}$. Let $r_j$ be the first zero entry in $V_W$ where $r_{j-1}>0$, if there is no such entry, then $j = k$. Consider the subvector $(r_i, r_{i+1}, ..., r_{j-1})$ of $V_W$ where all entries are positive. Since  row ${v_{i+1}}^{in}$ must dominate column $x_{v_i v_{i+1}}$, we have ${x^*}^{in}_{v_{i+1}} > 0$ and ${x^*}^{out}_{v_{i+1}} = 0$. Row $v_{i+1} v_{i+2}$ cannot dominate column $x_{v_{i+1} v_{i+2}}$ because edge $v_{i+1} v_{i+2}$ still has at least $r_{i+1}$ remaining capacity. Row ${v_{i+1}}^{out}$ cannot dominate column $x_{v_{i+1} v_{i+2}}$ because column $x^{in}_{v_{i+1}}$ is the least preferred and ${x^*}^{in}_{v_{i+1}} > 0$. Therefore, row ${v_{i+2}}^{in}$ must dominate column $x_{v_{i+1} v_{i+2}}$. By repeating analogous arguement, we eventually have row ${v_j}^{in}$ must dominate column $x_{v_{j-1} v_j}$.

Consider the following cases of $j$:
\begin{enumerate}
\item If $r_{j-1} > 0$ and $r_j=0$, then the only way to let $f_{in}(v_j)$ increase and $f_{out}(v_j)$ remain the same is $g_{v_j}$ follows equation~\ref{eq2} and $f(v_j v_{j+1})=0$. Therefore, row ${v_j}^{out}$ dominates column $x^{in}_{v_j}$ and $x^{out}_{v_j}$, ${x^*}^{in}_{v_j}=0$, and ${x^*}^{out}_{v_j}=M(v_j)$. This indicates row ${v_j}^{in}$ does not dominate column $x_{v_{j-1} v_j}$ because ${x^*}^{out}_{v_j}>0$ and $x^{out}_{v_j}$ is the least preferred column.

\item If $j = k$, by point 3 in Definition~\ref{def2.5}, if $v_k=t$ then row ${v_k}^{in}$ is not defined because there are no corresponding rows in $A$ for vertex $s$; if $v_{k-1} v_k \succ_{v_k} u v_k$ for some other edge $u v_k \in E$ and $f(u v_k) > 0$, then row ${v_k}^{in}$ does not dominate column $x_{v_{k-1} v_k}$ because it prefers column $x_{u v_k}$ more.
\end{enumerate}

From the aforementioned cases, either row ${v_j}^{in}$ is not defined or ${v_j}^{in}$ does not dominate column $x_{v_{j-1} v_j}$. We have a contradiction. Thus, by assigning $f(e) = {x^*}_e$ for each $e \in E$, feasibility and stability are guaranteed, we obtain a TL-SF.

TL-SF $\to$ Scarf:

Suppose $f$ is stable, first set ${x^*}_e = f(e)$ for each $e \in E$. If ${x^*}_e = c(e)$ then row $e$ dominates column $x_e$. The remaining is to assign ${x^*}^{in}_v$ and ${x^*}^{out}_v$ such that there are rows dominating column $x^{in}_v$, column $x^{out}_v$ where $v \in V$, and column $x_e$ where $e \in E$ is not saturated.

For $v_1 \in V$, if $v_1=s$ or there exists $v_2$ and $u$ both in $V$ such that $v_1 v_2 \succ_{v_1} v_1 u$, $f(v_1 u) > 0$, and $v_1 v_2$ is unsaturated, then column $x_{v_1 v_2}$ cannot be dominated by row $v_1 v_2$ and row ${v_1}^{out}$. This forces row ${v_2}^{in}$ to dominate column $x_{v_1 v_2}$. We should set ${x^*}^{in}_{v_2} = b_{{v_2}^{in}} - a_{v_2} \sum_{u v_2 \in E}{{x^*}_{uv_2}}$ and ${x^*}^{out}_{v_2} = 0$. Note that $v_1 v_2$ satisfies point 2 of Definition~\ref{def2.5}.

For $v_2$ and $v_3$ both in $V$, if $g_{v_2}$ follows equation~\ref{eq1}, $f_{in}(v_2)=0$ and $f(v_2 v_3) < b_{v_2}$, then $b_{{v_2}^{in}}=M(v_2)$ and $b_{{v_2}^{out}}=M(v_2)+b_{v_2}$. The only way to dominate column $x^{in}_{v_2}$ is to set ${x^*}^{in}_{v_2} = M(v_2)$ and ${x^*}^{out}_{v_2}=0$ so that row ${v_2}^{in}$ dominates column $x^{in}_{v_2}$.

We say vertices $v_2, v_2, ..., v_k$ are {\em inflow-dominated} if:

\begin{enumerate}
\item There is a walk $W = (v_1, v_2, ..., v_k)$ where $f(v_i v_{i+1}) < c(v_i v_{i+1})$ for $i=1, ..., k-1$.
\item Either one of the cases hold:
	\begin{enumerate}
		\item $v_1 v_2$ satisfies point 2 of Definition~\ref{def2.5}
		\item $f_{in}(v_2)=0$ and $f(v_2 v_3) < c(v_2 v_3)$.
	\end{enumerate}
\end{enumerate}

We know that ${x^*}^{in}_{v_2} = b_{{v_2}^{in}} - a_{v_2} \sum_{uv_2 \in E}{{x^*}_{uv_2}}$ ($\sum_{uv_2 \in E}{{x^*}_{uv_2}} = 0$ in case 2.(b)) and ${x^*}^{out}_{v_2} = 0$. Since $v_2 v_3$ is unsaturated, row $v_2 v_3$ cannot dominate column $x_{v_2 v_3}$. Row ${v_2}^{out}$ cannot dominate colulmn $x_{v_2 v_3}$ because ${x^*}^{in}_{v_2} > 0$ and column $x^{in}_{v_2}$ is the least preferred. This forces row ${v_3}^{in}$ to dominate column $x_{v_2 v_3}$. We should set ${x^*}^{in}_{v_3} = b_{{v_3}^{in}} - a_{v_3} \sum_{u v_3 \in E}{{x^*}_{u v_3}}$ and ${x^*}^{out}_{v_3} = 0$. By repeating anologous reasoning, for any inflow-dominated vertex $v_i$ where $i=2, ..., k$, we should set ${x^*}^{in}_{v_i} = b_{{v_i}^{in}} - a_{v_i} \sum_{uv_i \in E}{{x^*}_{uv_i}}$ and ${x^*}^{out}_{v_i} = 0$.

Similarly, we say vertices $v_1, v_2, ..., v_{k-1}$ are {\em outflow-dominated} if:

\begin{enumerate}
\item There is a walk $W = (v_1, v_2, ..., v_k)$ where $f(v_i v_{i+1}) < c(v_i v_{i+1})$ for $i=1, ..., k-1$.
\item Either one of the cases hold:
	\begin{enumerate}
		\item $v_{k-1} v_k$ satisfies point 3 of Definition~\ref{def2.5}
		\item $f_{out}(v_{k-1})=0$ and $f(v_{k-2} v_{k-1}) < c(v_{k-2} v_{k-1})$.
	\end{enumerate}
\end{enumerate}

By similar reasoning, for any outflow-dominated vertex $v_i$ where $i=1, ..., {k-1}$, we should set ${x^*}^{in}_{v_i} = 0$ and ${x^*}^{out}_{v_i} = b_{{v_i}^{out}} - \sum_{v_iu \in E}{{x^*}_{v_iu}}$.

No vertex in $V$ can be both inflow-dominated and outflow-dominated, otherwise there is a blocking walk which contradicts to Definition~\ref{def2.5}.

For vertex $v \in V$ that is neither inflow-dominated nor outflow-dominated, it is guaranteed that column $x_e$ where $e$ has $v$ as an endpoint is already dominated by some rows. We can arbitrarily assign non-negative values for ${x^*}^{in}_v$ and ${x^*}^{out}_v$ such that ${x^*}^{in}_v + {x^*}^{out}_v = \min\{b_{v^{in}} - a_v \sum_{uv \in E}{x^*}_{uv}, b_{v^{out}} - a_v \sum_{uv \in E}{x^*}_{uv}\}$.
\qed

\section{Proof of Theorem~\ref{thm5.1}}\label{app:thm5.1}

\begin{proof}
TL-Scarf $\to$ 3-layer-linear-SF:

Suppose $x^*$ is a solution of TL-Scarf, then $x^*$ dominates every column of $A$. We set $f'$ as the following:
\begin{enumerate}
\item For $e = uv \in E$, if $v \neq t$ set $f'(u_{out} m_e) = {x^*}_e$ and $f'(m_e v_{in}) = A_{v^{in} x_e}{x^*}_e$; if $v=t$ set $f'(u_{out} m_e) = f'(m_e v_{in}) = {x^*}_e$.
\item For each $v \in V$, set $f'(v_{out} m^{in}_v) = f'(m^{in}_v v_{in}) = {x^*}^{in}_v$ and $f'(v_{out} m^{out}_v) = f'(m^{out}_v v_{in}) = {x^*}^{out}_v$.
\item Set $f'(s' v_{out}) = \sum_{e = vu \in E}A_{v^{out} x_e}{x^*}_e + {x^*}^{in}_v + {x^*}^{out}_v = \sum_{e = vu \in E}{x^*}_e + {x^*}^{in}_v + {x^*}^{out}_v$.
\item Set $f'(v_{in} t') = \sum_{e = uv \in E}A_{v^{in} x_e}{x^*}_e + {x^*}^{in}_v + {x^*}^{out}_v$.
\item Set $f'(s' s_{out}) = \sum_{e = sv \in E}A_{v^{out} x_e}{x^*}_e = \sum_{e = sv \in E}{x^*}_e$.
\item Set $f'(t_{in} t') = \sum_{e = vt \in E}{x^*}_e$.
\end{enumerate}

We can see that feasibility is satisfied in the three-layer linear-network. The remaining is to show that $f'$ is stable. Assume there is a blocking walk $W$ in this network. Observing the structure of this three-layer network, each vertex $m_e$, $m^{in}_v$, or $m^{out}_v$ in the second layer only have one incoming and one outgoing edge. This indicates they cannot be the starting or ending vertex of the blocking walk. Besides, $W$ cannot start from $s'$ and end at a vertex in the first layer since each vertex in that layer has only one incoming edge. Similarly, $W$ cannot start from a vertex in the third layer and end at $t'$ as each vertex in the third layer has only one outgoing edge. The followings are the remaining cases of $W$. We show a proof by contradiction for each case.

\begin{enumerate}
\item $W$ starts from $s'$ and ends at a vertex $v_{in}$ in the third layer:

This means there exists a vertex $u_{out}$ in the first layer such that row $u^{out}$ does not bind since $s' u_{out}$ is not saturated. For $e=uv \in E$, $m_e v_{in}$ and $u_{out} m_e$ are not saturated and $v_{in}$ prefers $m_e v_{in}$ to some other nonzero incoming edges. This indicates column $x_e$ is dominated by neither row $e$ nor row $v^{in}$. Column $x_e$ is not dominated by $x^*$, a contradiction.

\item $W$ starts from a vertex $u_{out}$ in the first layer and ends at $t'$:

This means there exists a vertex $v_{in}$ in the third layer such that row $v^{in}$ does not bind since $v_{in} t'$ is not saturated. For $e=uv \in E$, $m_e v_{in}$ and $u_{out} m_e$ are not saturated and $u_{out}$ prefers $u_{out} m_e$ to some other nonzero outgoing edges. This indicates column $x_e$ is dominated by neither row $e$ nor row $u^{out}$. Column $x_e$ is not dominated by $x^*$, a contradiction.

\item $W$ starts from a vertex $u_{out}$ in the first layer and ends at a vertex $v_{in}$ in the third layer:

This means for $e=uv \in E$, $m_e v_{in}$ and $u_{out} m_e$ are not saturated so column $x_e$ is not dominated by row $e$. $u_{out}$ prefers $u_{out} m_e$ to some other nonzero outgoing edges. This indicates column $x_e$ is not dominated by row $u^{out}$. $v_{in}$ prefers $m_e v_{in}$ to some other nonzero incoming edges. This indicates column $x_e$ is not dominated by row $v^{in}$. Column $x_e$ is not dominated by $x^*$, a contradiction.

\item $W$ starts from $s'$ to $t'$:

This means there are some rows $v^{in}$ and $u^{out}$ that do not bind and for $e=uv \in E$, $m_e v_{in}$ and $u_{out} m_e$ are not saturated which indicates $x_e$ is not dominated by row $e$. Column $x_e$ is not dominated by $x^*$, a contradiction.

\end{enumerate}

3-layer-linear-SF $\to$ TL-Scarf:

Suppose $f'$ is a stable flow of the three-layer linear network. For $e=uv \in E$, if $m_e v_{in}$ and $u_{out} m_e$ are saturated, then ${x^*}_e = c'(u_{out} m_e) = b_e$. Column $x_e$ is dominated by row $e$. Otherwise, the following two cases cannot happen at the same time or there is a blocking walk $(u_{out}, m_e, v_{in})$.

\begin{enumerate}
\item $u_{out} = s_{out}$ or $u_{out}$ prefers $u_{out} m_e$ to some other nonzero outgoing edges.
\item $v_{in} = t_{in}$ or $v_{in}$ prefers $m_e v_{in}$ to some other nonzero incoming edges.
\end{enumerate}

If the first case happens, $v_{in}$ must prefer all nonzero incoming edges to $m_e v_{in}$ and $v_{in} t'$ must be saturated or the walk $(u_{out}, m_e, v_{in}, t')$ is blocking. This forces $f'(m^{in}_v v_{in}) = c'(v_{in} t') - \sum_{e' = wv_{in} \in E'}f(e')$ which corresponds to row $v^{in}$ binds and dominates column $x_e$.

If the second case happens, $u_{out}$ must prefer all nonzero outgoing edges to $u_{out} m_e$ and $s' u_{out}$ must be saturated or the walk $(s', u_{out}, m_e, v_{in})$ is blocking. This forces $f(u_{out} m^{out}_v) = c'(s' u_{out}) - \sum_{e' = u_{out}w \in E'} f(e')$ which corresponds to row $u^{out}$ binds and dominates column $x_e$.

The remaining is to show how $f'$ makes column $x^{in}_v$ and $x^{out}_v$ dominated for each $v \in V$. If both $s' v_{out}$ and $v_{in} t'$ are saturated, then both row $v^{in}$ and $v^{out}$ bind. We know that both walks $(v_{out}, m^{in}_v, v_{in})$ and $(v_{out}, m^{out}_v, v_{in})$ are not blocking. Therefore, both column $x^{in}_v$ and $x^{out}_v$ are dominated. It is not possible that both $s' v_{out}$ and $v_{in} t'$ are not saturated, otherwise $(s', v_{out}, m^{in}_v, v_{in}, t')$ or $(s', v_{out}, m^{out}_v, v_{in}, t')$ is a blocking walk. The remaining case is exactly one of $s' v_{out}$ and $v_{in} t'$ is saturated. 

If $s' v_{out}$ is not saturated, then by the setting of $c'(m^{in}_v v_{in})$ and $c'(v_{out} m^{in}_v)$, $m^{in}_v v_{in}$ and $v_{out} m^{out}_v$ cannot be saturated, this forces all the incoming edges of $v_{in}$ except $m^{in}_v v_{in}$ have to be zero. Otherwise, $(s, v_{out}, m^{in}_v, v_{in})$ is a blocking walk since $v_{in}$ prefers $m^1_v v_{in}$ the most. This also forces $v_{in} t'$ saturated otherwise $(s, v_{out}, m^{in}_v, v_{in}, t)$ is blocking. Consequently, $f(v_{out} m^{in}_v) = f(m^{in}_v v_{in}) = c'(v_{in} t')$ corresponds to row $v^{in}$ binds and dominates column $x^{in}_v$ and $x^{out}_v$. This can only happen when $c'(s' v_{out}) > c'(v_{in} t')$, i.e. $b_{v^{out}} > b_{v^{in}}$. Similarly, if $v_{in} t'$ is not saturated, then $f(v_{out} m^{in}_v) = f(m^{in}_v v_{in}) = c'(s' v_{out})$ corresponds to row $v^{out}$ binds and dominates column $x^{in}_v$ and $x^{out}_v$. This can only happen when $c'(s' v_{out}) < c'(v_{in} t')$, i.e. $b_{v^{out}} < b_{v^{in}}$.
\qed
\end{proof}

\section{Proof of Corrolary~\ref{cor5.1}}\label{app:cor5.1}

This is a direct result by combining section~\ref{subsubsec3.2.1} and Theorem~\ref{thm5.1}. We will use the same notaion as in section~\ref{subsubsec3.2.1} for the TL-network instance.

Given the TL-network $(G,s,t,c)$, $G=(V\cup\{s,t\},E)$, and $\succ_v$ for $v \in V$, construct the three-layer AL-network with $(G',s',t',c')$, $G'=(V'\cup\{s',t'\},E')$, and $\succ_{v'}$ for $v' \in V'$ as the following:

\begin{enumerate}
\item Create $s_{out} \in V'$ and $t_{in} \in V'$.
\item For each $v \in V$, create $v_{in}$, $v_{out}$, $m^{in}_v$, and $m^{out}_v$ in $V'$.
\item Let $M(v) = \max(\sum_{uv \in E} c(uv), \sum_{vu \in E} c(vu), b_v) + 1$.
\item If $g_v$ is in the form of equation~\ref{eq1}, set $c'(s' v_{in}) = M(v)$ and $c'(s' v_{out}) = M(v) + b_v$.
\item If $g_v$ is in the form of equation~\ref{eq2}, set $c'(s' v_{in}) = M(v) + b_v$ and $c'(s' v_{out}) = M(v)$.
\item Set $c'(s' s_{out}) = \sum_{su \in E} {c(su)} + 1$ and $c'(t_{in} t') = \sum_{ut \in E} {c(ut)} + 1$.
\item For each $e = uv \in E$, create $m_e \in E'$ and set $c'(u_{out} m_e) = c(u v)$ and $c'(m_e v_{in}) = a_v c(uv)$.
\item Set $c'(v_{out} m^{in}_v) = c'(v_{out} m^{out}_v) = c'(m^{in}_v v_{in}) = c'(m^{out}_v v_{in}) = M(v) + b_v$.
\item For each $v_{out} \in V'$, $v_{out}$ prefers $v_{out} m^{out}_v$ the most and $v_{out} m^{in}_v$ the least, the preference of $v_{out} m_e$ for some $e = vu \in E$ is the same as $\succ_v$.
\item For each $v_{in} \in V'$, $v_{in}$ prefers $m^{in}_v v_{in}$ the most and $m^{out}_v v_{in}$ the least, the preference of $m_e v_{in}$ for some $e = uv \in E$ is the same as $\succ_v$.
\item The preference of $s_{out}$ and $t_{in}$ is arbitrary.
\item Set $g_{s_{out}}(x)=x$ and $g_{t_{in}}(x)=x$. For each $v \in V$, set $g_{v_{in}}(x)=x$, $g_{v_{out}}(x)=x$, $g_{m^1_v}(x)=x$, and $g_{m^2_v}(x)=x$.
\item For each $e=uv \in E$, set $g_{m_e}(x) = a_v x$ if $v \neq t'$. If $v=t$, set $g_{m_e}(x)=x$.
\end{enumerate}

Suppose $f'$ is an AL-SF of the three-layer AL-network, to obtain the TL-SF $f$ for the old TL-network, for each $e=uv \in E$, set $f(uv) = f'(u_{out} m_e)$.
\qed

\section{Examples}
\begin{example} \label{ex2}
Consider the following TL-network. On each edge, the upper number  represents the flow value and the lower number correspond to the capacity.
\mbox{}\\
\begin{center}
\begin{tikzpicture}
	\node[vertex] (a1) at (0,0) {$s$}; 
	\node[vertex] (b1) at (3,1.5) {$u$}; 
	\node[vertex] (c1) at (3,-1.5) {$v$};
  \node[vertex] (d1) at (6,0) {$t$};
	\path[->]
		(a1) edge node [above] {$1/1$} (b1)
		(b1) edge  [bend right=45,looseness=0.5] node [left] {$0/1$} (c1)
		(c1) edge  [bend right=45,looseness=0.5] node [right] {$0/2$} (b1)
        (b1) edge node [above] {$1/11$} (d1);
	\node at (3,2.5) {$uv \succ_u ut$ and $vu \succ_u su$};		\node at (10,1.5) {$g_u(x) = \begin{cases} 
      0 & \text{if } x<0.5, \\
      2 x - 1 & \text{otherwise.} \\
   \end{cases}$};
	\node at (10,-1) {$g_v(x) = \begin{cases} 
      [0,2] & \text{if } x=0, \\
      3x + 2 & \text{otherwise.} \\
   \end{cases}$};
\end{tikzpicture}
\end{center}

$u$ prefers any incoming or outgoing edges that have $v$ as an endpoint. The flow assignment given in the picture is feasible but not stable.   There are two blocking walks: $W_1=(u,v,u)$ and $W_2=(u,v,u,t)$. There exists a vector $V_{W_1}=(0,2)$ that can be rerouted along $W_1$. There exists a vector $V_{W_2}=(0,2,4)$ that can be rerouted along $W_2$.

The following flow assignment is stable:
\mbox{}\\
\begin{center}
\begin{tikzpicture}
	\node[vertex] (a1) at (0,0) {$s$}; 
	\node[vertex] (b1) at (3,1.5) {$u$}; 
	\node[vertex] (c1) at (3,-1.5) {$v$};
  \node[vertex] (d1) at (6,0) {$t$};
	\path[->]
		(a1) edge node [above] {$1/1$} (b1)
		(b1) edge  [bend right=45,looseness=0.5] node [left] {$0/1$} (c1)
		(c1) edge  [bend right=45,looseness=0.5] node [right] {$2/2$} (b1)
        (b1) edge node [above] {$5/11$} (d1);
	\node at (3,2.5) {$uv \succ_u ut$ and $vu \succ_u su$};		\node at (10,1.5) {$g_u(x) = \begin{cases} 
      0 & \text{if } x<0.5, \\
      2 x - 1 & \text{otherwise.} \\
   \end{cases}$};
	\node at (10,-1) {$g_v(x) = \begin{cases} 
      [0,2] & \text{if } x=0, \\
      3x + 2 & \text{otherwise.} \\
   \end{cases}$};
\end{tikzpicture}
\end{center}

\end{example}

\begin{example} \label{ex3}
Consider the following TL-network:
\mbox{}\\
\begin{center}
\begin{tikzpicture}
	\node[vertex] (s) at (0,0) {$s$}; 
	\node[vertex] (u) at (4,0) {$u$}; 
	\node[vertex] (v) at (6,1.5) {$v$};
	\node[vertex] (w) at (8,0) {$w$};
    \node[vertex] (t) at (12,0) {$t$};
	\path[->]
		(s) edge node [above] {$0/3$} (u)
		(u) edge node [left] {$1/3$} (v)
		(v) edge node [right] {$1/3$} (w)
		(u) edge node [above] {$2/3$} (w)
		(w) edge node [above] {$10/10$} (t)
        ;
	\node at (4,-1) {$uv \succ_u uw$};
	\node at (8,-1) {$uw \succ_w vw$};
	\node at (6,-2.5) {$g_u(x) = \begin{cases} 
      [0,4] & \text{if } x=0, \\
      x + 4 & \text{otherwise.} \\
   \end{cases}$,
$g_v(x) = \begin{cases} 
      0 & \text{if } x<0.5, \\
      2 x - 1 & \text{otherwise.} \\
   \end{cases}$,
$g_w(x) = \begin{cases} 
      [0,1] & \text{if } x=0, \\
      3x + 1 & \text{otherwise.} \\
   \end{cases}$};
\end{tikzpicture}
\end{center}

The corresponding Scarf's solution is the following where the blanks are zeros:
{\small
\[
A x^* = 
\begin{blockarray}{cccccccccccc}
& x_{su} & x_{uv} & x_{vw} & x_{uw} & x_{wt} & x^{in}_u & x^{out}_u & x^{in}_v & x^{out}_v & x^{in}_w & x^{out}_w \\ 
\begin{block}{c(ccccc|cc|cc|cc)}
su & 1 &  &  &  &  &  &  &  &  &  &  \\ 
uv &  & 1 &  &  &  &  &  &  &  &  &  \\
vw &  &  & 1 &  &  &  &  &  &  &  &  \\
uw &  &  &  & 1 &  &  &  &  &  &  &  \\
wt &  &  &  &  & 1 &  &  &  &  &  &  \\
\cline{2-12}
u^{in} & 1 &  &  &  &  & 1 & 1 &  &  &  &  \\
u^{out} &  & 1 &  & 1 &  & 1 & 1 &  &  &  &  \\ 
\cline{2-12}
v^{in} &  & 2 &  &  &  &  &  & 1 & 1 &  &  \\
v^{out} &  &  & 1 &  &  &  &  & 1 & 1 &  &  \\
\cline{2-12}
w^{in} &  &  & 3 & 3 &  &  & & &  & 1 & 1 \\
w^{out} &  &  &  &  & 1 &  & & &  & 1 & 1 \\
\end{block}
\end{blockarray}
\quad
\begin{blockarray}{c}
\\ 
\begin{block}{(c)}
\textcolor{red}{0} \\ \textcolor{red}{1} \\ \textcolor{red}{1} \\ \textcolor{red}{2} \\ \textcolor{red}{10} \\ 7 \\ 0 \\ 3 \\ 0 \\ 2 \\ 0 \\
\end{block}
\end{blockarray}
\leqslant
\begin{blockarray}{c}
\\ 
\begin{block}{(c)}
3 \\ 3 \\ 3 \\ 3 \\ 10 \\ 7 \\ 11 \\ 5 \\ 4 \\ 11 \\ 12 \\
\end{block}
\end{blockarray}
=
\begin{blockarray}{c}
\\ 
\begin{block}{(c)}
b_{su} \\ b_{uv} \\ b_{vw} \\ b_{uw} \\ b_{wt} \\ b_{u^{in}} \\ b_{u^{out}} \\ b_{v^{in}} \\ b_{v^{out}} \\ b_{w^{in}} \\ b_{w^{out}} \\
\end{block}
\end{blockarray}
\]
}
The following table is the rows ranking over columns from the highest to the lowest:
\mbox{}\\

\begin{tabular}{ | l | l | l | l |}
    \hline
    row & $u^{in}$ & $u^{out}$ & $v^{in}$ \\ \hline
    column & $x^{in}_u, x_{su}, x^{out}_u$
    	& $x^{out}_u, x_{uv}, x_{uw}, x^{in}_u$
        & $x^{in}_v, x_{uv}, x^{out}_v$
\\ \hline
  	\end{tabular}
\mbox{}\\

\begin{tabular}{ | l | l | l | l |}
    \hline
    row  & $v^{out}$ & $w^{in}$ & $w^{out}$ \\ \hline
    column & $x^{out}_v, x_{vw}, x^{in}_v$
        & $x^{in}_w, x_{uw}, x_{vw}, x^{out}_w$
        & $x^{out}_w, x_{wt}, x^{in}_w$
\\ \hline
  	\end{tabular}
\mbox{}\\

By the aforementioned construction, $M(u)=7$, $M(v)=4$, and $M(w)=11$. Clearly, the flow capcity constraints are satisfied and only edge $wt$ is saturated which corresponds to the only capacity binding constraint row $wt$. Flow feasibility is also satisfied. We can see that row $u^{out}$ does not bind and $f_{out}(u)=3 < 4$ does not reach the truncated threshold. For other vertices, the inflow and outflow constraints all bind.

For stability, since row $su$ does not dominate column $x_{su}$, row $u^{in}$ is the only choice to dominate column $x_{su}$. By the setting of $M(u)$, one of ${x^*}^{in}_u$ and ${x^*}^{out}_u$ must be positive, this forces ${x^*}^{in}_u = 7$ and ${x^*}^{out}_u = 0$. Since row $u^{out}$ prefers column $x^{in}_u$ the least, row $v^{in}$ must dominate column $x_{uv}$ and row $w^{in}$ must dominate column $x_{uw}$. Therefore, ${x^*}^{in}_v > 0$ and ${x^*}^{out}_v = 0$, and ${x^*}^{in}_w > 0$ and ${x^*}^{out}_w = 0$. Column $x_{vw}$ is less preferred than column $x_{uw}$ and column $x^{in}_w$ by row $w^{in}$. This ensures that $x^*$ dominates all the columns of $A$.

\end{example}

\begin{example} \label{ex4}
\mbox{}\\
\begin{center}
\begin{tikzpicture}
      \draw[->] (-0.7*1.3,2.5*1.3) -- (3*1.3,2.5*1.3) node[right] {$x$};
      \draw[->] (-0.5*1.3,2.3*1.3) -- (-0.5*1.3,5.5*1.3) node[above] {$g_v(x)$};
      \node at (0.5*1.3,2.3*1.3) {$c_1$};
      \node at (1.5*1.3,2.3*1.3) {$c_2$};
      \node at (0,5.2*1.3) {$g_{v,1}$};
      \node at (1*1.3,5.2*1.3) {$g_{v,2}$};
      \node at (2*1.3,5.2*1.3) {$g_{v,3}$};
      \draw[domain=0:0.5,smooth,variable=\x,blue] plot (-0.5*1.3,{(\x+2.5)*1.3}) [domain=-0.5:0.5,smooth,variable=\x,blue] plot({\x*1.3},{(0.25*(\x+0.5)+0.5+2.5)*1.3});
    \draw[domain=0.5:1.5,smooth,variable=\x,blue] plot({\x*1.3},{(1.2*(\x-0.5)+0.75+2.5)*1.3});
    \draw[domain=1.5:2.2,smooth,variable=\x,blue] plot({\x*1.3},{(1.95+2.5)*1.3})
[domain=2.2:3,smooth,variable=\x,blue] plot({\x*1.3},{(0.7*(\x-2.2)+1.95+2.5)*1.3});
	  \draw[domain=0:3,dotted,variable=\x,red] plot (0.5*1.3,{(\x+2.5)*1.3});
      \draw[domain=0:3,dotted,variable=\x,red] plot (1.5*1.3,{(\x+2.5)*1.3});
      \node at (5*1.3,5.3*1.3) {PL-network:};
      \node [vertex] (v) at (7*1.3,4*1.3) {$v$};
      \node [vertex] (a) at (5*1.3,4.5*1.3) {$u$};
      \node [vertex] (b) at (5*1.3,3.5*1.3) {$w$};
      \node [vertex] (c) at (9*1.3,4.5*1.3) {$y$};
      \node [vertex] (d) at (9*1.3,3.5*1.3) {$z$};
      \path[->]
    	(a) edge node [above, pos=0.85] {$1^{\mbox{st}}$} (v)
        (b) edge node [below, pos=0.85] {$2^{\mbox{nd}}$} (v)
        (v) edge node [above, pos=0.15] {$1^{\mbox{st}}$} (c)
        (v) edge node [below, pos=0.15] {$2^{\mbox{nd}}$} (d)
        ;
      \node at (-1,1*1.3-1) {TL-network:};
{\footnotesize
      \node[vertex] (vin) at (3*1.3,1*1.3-1) {$v_{in}$};
      \node[vertex] (vout) at (7*1.3,1*1.3-1) {$v_{out}$};
      \node[vertex] (aout) at (1*1.3,1.5*1.3-1) {$u_{out}$};
      \node[vertex] (bout) at (1*1.3,0.5*1.3-1) {$w_{out}$};
      \node[vertex] (cin) at (9*1.3,1.5*1.3-1) {$y_{in}$};
      \node[vertex] (din) at (9*1.3,0.5*1.3-1) {$z_{in}$};
      \node[vertex] (v1) at (5*1.3,2*1.3-1) {$v_1$};
      \node[vertex] (v2) at (5*1.3,1*1.3-1) {$v_2$};
      \node[vertex] (v3) at (5*1.3,0-1) {$v_3$};
      \path[->]
      	(aout) edge node [above, pos=0.9] {$1^{\mbox{st}}$} (vin)
        (bout) edge node [below, pos=0.9] {$2^{\mbox{nd}}$} (vin)
        (vout) edge node [above, pos=0.1] {$1^{\mbox{st}}$} (cin)
        (vout) edge node [below, pos=0.1] {$2^{\mbox{nd}}$} (din)
      ;
      \node at (5*1.3,2.5*1.3-1) {$g_{v,1}$};
      \node at (5*1.3,1.5*1.3-1) {$g_{v,2}$};
      \node at (5*1.3,0.5*1.3-1) {$g_{v,3}$};
      \path[->]
    	(vin) edge [bend left=45,looseness=0.5] node [above, pos=0.15] {$1^{\mbox{st}}$} (v1)
        (vin) edge node [above, pos=0.15] {$2^{\mbox{nd}}$} (v2)
        (vin) edge [bend right=45,looseness=0.5] node [below, pos=0.15] {$3^{\mbox{rd}}$} (v3)
        (v1) edge [bend left=45,looseness=0.5] node [above, pos=0.85] {$1^{\mbox{st}}$} (vout)
        (v2) edge node [above, pos=0.85] {$2^{\mbox{nd}}$} (vout)
        (v3) edge [bend right=45,looseness=0.5] node [below, pos=0.85] {$3^{\mbox{rd}}$} (vout)
        ;
}
    \end{tikzpicture}
\end{center}

In this example, we focus on vertex $v \in V$. $uv$, $wv$, $vy$, and $vz$ are all in $E$. We split each vertex and connect the out vertices to the in vertices, so $u_{out} v_{in}$, $w_{out} v_{in}$, $v_{out} y_{in}$, and $v_{out} z_{in}$ are in the new TL-network. In the original PL-network, $uv \succ_v wv$ and $vy \succ_v vz$, so $u_{out} v_{in} \succ_{v_{in}} w_{out} v_{in}$ and $v_{out} y_{in} \succ_{v_{out}} v_{out} z_{in}$ in the new TL-network. Besides, we set proper capacities and prioritize the preference for $v_{in} v_i$ and $v_i v_{out}$ for $i=1,2,3$ such that $f_{out}(v)$ meets the expected flow value $g_v(f_{in}(v))$.
\end{example}

\begin{example} \label{ex9}
Consider the following TL-network with its stable flow assignment:
\mbox{}\\
\begin{center}
\begin{tikzpicture}
	\node[vertex] (a1) at (0,0) {$s$}; 
	\node[vertex] (b1) at (3,1.5) {$u$}; 
	\node[vertex] (c1) at (3,-1.5) {$v$};
  \node[vertex] (d1) at (6,0) {$t$};
	\path[->]
		(a1) edge node [above] {$0/1$} (b1)
		(b1) edge  [bend right=45,looseness=0.5] node [left] {$2/3$} (c1)
		(c1) edge  [bend right=45,looseness=0.5] node [right] {$7/7$} (b1)
        (b1) edge node [above] {$11/11$} (d1);
	\node at (3,2.5) {$uv \succ_u ut$ and $vu \succ_u su$};		\node at (10,1.5) {$g_u(x) = \begin{cases} 
      0 & \text{if } x<0.5, \\
      2 x - 1 & \text{otherwise.} \\
   \end{cases}$};
	\node at (10,-1) {$g_v(x) = \begin{cases} 
      [0,1] & \text{if } x=0, \\
      3x + 1 & \text{otherwise.} \\
   \end{cases}$};
\end{tikzpicture}
\end{center}

By the construction in section~\ref{subsubsec3.2.1}, we have the following TL-Scarf:

\[
A x^* = 
\begin{blockarray}{ccccccccc}
& x_{su} & x_{uv} & x_{vu} & x_{ut} & x^{in}_u & x^{out}_u & x^{in}_v & x^{out}_v \\ 
\begin{block}{c(cccc|cc|cc)}
su & 1 &  &  &  &  &  &  &  \\ 
uv &  & 1 &  &  &  &  &  &  \\
vu &  &  & 1 &  &  &  &  &  \\
ut &  &  &  & 1 &  &  &  &  \\
\cline{2-9}
u^{in} & 2 &  & 2 &  & 1 & 1 &  &  \\
u^{out} &  & 1 &  & 1 & 1 & 1 &  &  \\ 
\cline{2-9}
v^{in} &  & 3 &  &  &  &  & 1 & 1 \\
v^{out} &  &  & 1 &  &  &  & 1 & 1 \\
\end{block}
\end{blockarray}
\quad
\begin{blockarray}{c}
\\ 
\begin{block}{(c)}
0 \\ 2 \\ 7 \\ 11 \\ 2 \\ 0 \\ 2 \\ 0 \\ 
\end{block}
\end{blockarray}
=
\begin{blockarray}{c}
\\ 
\begin{block}{(c)}
0 \\ 2 \\ 7 \\ 11 \\ 16 \\ 15 \\ 8 \\ 9 \\
\end{block}
\end{blockarray}
\leqslant
\begin{blockarray}{c}
\\ 
\begin{block}{(c)}
1 \\ 3 \\ 7 \\ 11 \\ 16 \\ 15 \\ 8 \\ 9 \\
\end{block}
\end{blockarray}
=
\begin{blockarray}{c}
\\ 
\begin{block}{(c)}
b_{su} \\ b_{uv} \\ b_{vu} \\ b_{ut} \\ b_{u^{in}} \\ b_{u^{out}} \\ b_{v^{in}} \\ b_{v^{out}} \\
\end{block}
\end{blockarray}
\]

Where the preference of each row over columns from the highest to the lowest is:
\mbox{}\\

\begin{tabular}{ | l | l | l | l | l |}
    \hline
    row & $u^{in}$ & $u^{out}$ & $v^{in}$ & $v^{out}$ \\ \hline
    column & $x^{in}_u, x_{vu} x_{su}, x^{out}_u$
    	& $x^{out}_u, x_{uv}, x_{ut}, x^{in}_u$
        & $x^{in}_v, x_{uv}, x^{out}_v$
			& $x^{out}_v, x_{vu}, x^{in}_v$
\\ \hline
  	\end{tabular}
\mbox{}\\

By the aforementioned construction, we obtain the corresponding three-layer linear network and the preference of each vertex as the following:

\mbox{}\\
\begin{tabular}{ | l | l |}
    \hline
    vertex & preference \\ \hline
    $u_{out}$ & $u_{out}m^{out}_u \succ_{u_{out}} u_{out}m_{uv} \succ_{u_{out}} u_{out}m_{ut} \succ_{u_{out}} u_{out}m^{in}_u$ \\ \hline
    $u_{in}$ & $m^{in}_u u_{in} \succ_{u_{in}} m_{vu}u_{in} \succ_{u_{in}} m_{su}u_{in} \succ_{u_{in}} m^{out}_u u_{in}$ \\ \hline
    $v_{out}$ & $v_{out}m^{out}_v \succ_{v_{out}} v_{out}m_{vu} \succ_{v_{out}} v_{out}m^{in}_v$ \\ \hline
    $v_{in}$ & $m^{in}_v v_{in} \succ_{v_{in}} m_{uv}v_{in} \succ_{v_{in}} m^{out}_v v_{in}$ \\ \hline
  	\end{tabular}

\mbox{}\\

\begin{center}
\begin{tikzpicture}
	\node[vertex] (s) at (0,0) {$s'$};
	\node[vertex] (sout) at (3,5) {$s_{out}$}; 
	\node[vertex] (uout) at (3,0) {$u_{out}$};
  \node[vertex] (vout) at (3,-5) {$v_{out}$};
	\node[vertex] (tin) at (10.5,5) {$t_{in}$}; 
	\node[vertex] (uin) at (10.5,0) {$u_{in}$};
  \node[vertex] (vin) at (10.5,-5) {$v_{in}$};
	\node[vertex] (t) at (13.5,0) {$t'$};
	\node[vertex] (msu) at (8,3.5) {$m_{su}$};
	\node[vertex] (muv) at (8,-2) {$m_{uv}$};
	\node[vertex] (mut) at (5.5,3.5) {$m_{ut}$};
	\node[vertex] (mvu) at (5.5,-2) {$m_{vu}$};
	\node[vertex] (m1u) at (6.75,2) {$m^{in}_u$};
	\node[vertex] (m2u) at (6.75,0) {$m^{out}_u$};
	\node[vertex] (m1v) at (6.75,-3) {$m^{in}_v$};
	\node[vertex] (m2v) at (6.75,-5) {$m^{out}_v$};
	\path[->]
		(s) edge node [left] {$0/2$} (sout)
		(s) edge node [above] {$15/15$} (uout)
		(s) edge node [left] {$9/9$} (vout)
		(sout) edge node [above] {$\textcolor{red}{0}/1$} (msu)
		(msu) edge node [right] {$0/2$} (uin)
		(uout) edge node [left] {$\textcolor{red}{11}/11$} (mut)
		(mut) edge node [above] {$11/11$} (tin)
		(uout) edge node [above] {$\textcolor{red}{2}/16$} (m1u)
		(m1u) edge node [above] {$2/16$} (uin)
		(uout) edge node [above] {$\textcolor{red}{0}/16$} (m2u)
		(m2u) edge node [above] {$0/16$} (uin)
		(uout) edge node [above] {$\textcolor{red}{2}/3$} (muv)
		(muv) edge node [right] {$6/9$} (vin)
		(vout) edge node [left] {$\textcolor{red}{7}/7$} (mvu)
		(mvu) edge node [above] {$14/14$} (uin)
		(vout) edge node [above] {$\textcolor{red}{2}/9$} (m1v)
		(m1v) edge node [above] {$2/9$} (vin)
		(vout) edge node [above] {$\textcolor{red}{0}/9$} (m2v)
		(m2v) edge node [above] {$0/9$} (vin)
		(tin) edge node [right] {$11/12$} (t)
		(uin) edge node [above] {$16/16$} (t)
		(vin) edge node [right] {$8/8$} (t);
\end{tikzpicture}
\end{center}

The linear mapping of each vertex $v$ in the three-layer linear network is:
\mbox{}\\

\begin{tabular}{ | l | l | l | l | l |}
    \hline
    $v$ & $m_{su}$ & $m_{vu}$ & $m_{uv}$ & others\\ \hline
    $g_v(x)$ & $2x$ & $2x$ & $3x$ & x\\ \hline
  	\end{tabular}
\mbox{}\\

Note that for vertices not listed in the table, the inflow is equal to the outflow.

We can see that the three-layer linear network captures the properties of the TL-Scarf. The capacities of edges adjacent to $s'$ and $t'$ are set to large values such that $s' s_{out}$ and $t_{in} t$ cannot be saturated, and for $u$ in the original TL-network, there must be some flow passing $m^{in}_u$ or $m^{out}_u$, similar arguement holds for $v$. The preference of each vertex simply follows the row preference in TL-Scarf. The linear mapping of each vertex is set in a way that follows the coefficient given in TL-Scarf.
\end{example}

\begin{example} \label{ex10}
Consider the TL-SF instance in Example~\ref{ex9}.

\mbox{}\\
\begin{center}
\begin{tikzpicture}
	\node[vertex] (a1) at (0,0) {$s$}; 
	\node[vertex] (b1) at (3,1.5) {$u$}; 
	\node[vertex] (c1) at (3,-1.5) {$v$};
  \node[vertex] (d1) at (6,0) {$t$};
	\path[->]
		(a1) edge node [above] {$0/1$} (b1)
		(b1) edge  [bend right=45,looseness=0.5] node [left] {$2/3$} (c1)
		(c1) edge  [bend right=45,looseness=0.5] node [right] {$7/7$} (b1)
        (b1) edge node [above] {$11/11$} (d1);
	\node at (3,2.5) {$uv \succ_u ut$ and $vu \succ_u su$};		\node at (10,1.5) {$g_u(x) = \begin{cases} 
      0 & \text{if } x<0.5, \\
      2 x - 1 & \text{otherwise.} \\
   \end{cases}$};
	\node at (10,-1) {$g_v(x) = \begin{cases} 
      [0,1] & \text{if } x=0, \\
      3x + 1 & \text{otherwise.} \\
   \end{cases}$};
\end{tikzpicture}
\end{center}

The left figure is the corresponding GSA instance obtained by Corollary~\ref{cor5.3} with edge capacities in fraction. The right figure is the GSA solution.
\mbox{}\\

\begin{center}
\begin{tikzpicture}
	\node[vertex] (is) at (0,3) {$i_s$}; 
	\node[vertex] (iu) at (0,0) {$i_u$};
  \node[vertex] (iv) at (0,-3) {$i_v$};
	\node[vertex] (jt) at (6,3) {$j_t$}; 
	\node[vertex] (ju) at (6,0) {$j_u$};
  \node[vertex] (jv) at (6,-3) {$j_v$};
	\path[-]
		(is) edge node [above, pos=0.3] {$1/2$} (ju)
		(iu) edge node [below, pos=0.3] {$3/15$} (jv)
		(iu) edge node [above, pos=0.7] {$11/15$} (jt)
		(iv) edge node [below, pos=0.7] {$7/9$} (ju);
	\path[->]
		(iu) edge [bend right=15,looseness=0.5] node [below] {$16/15$} (ju)
		(ju) edge [bend right=15,looseness=0.5] node [above] {$16/15$} (iu)
		(iv) edge [bend right=15,looseness=0.5] node [below] {$9/9$} (jv)
		(jv) edge [bend right=15,looseness=0.5] node [above] {$9/9$} (iv);
	\node[vertex] (iis) at (9,3) {$i_s$}; 
	\node[vertex] (iiu) at (9,0) {$i_u$};
  \node[vertex] (iiv) at (9,-3) {$i_v$};
	\node[vertex] (jjt) at (15,3) {$j_t$}; 
	\node[vertex] (jju) at (15,0) {$j_u$};
  \node[vertex] (jjv) at (15,-3) {$j_v$};
	\path[-]
		(iis) edge node [above, pos=0.3] {$\textcolor{red}{0}/2$} (jju)
		(iiu) edge node [below, pos=0.3] {$\textcolor{red}{2}/15$} (jjv)
		(iiu) edge node [above, pos=0.7] {$\textcolor{red}{11}/15$} (jjt)
		(iiv) edge node [below, pos=0.7] {$\textcolor{red}{7}/9$} (jju);
	\path[->]
		(iiu) edge [bend right=15,looseness=0.5] node [below] {$0/15$} (jju)
		(jju) edge [bend right=15,looseness=0.5] node [above] {$2/15$} (iiu)
		(iiv) edge [bend right=15,looseness=0.5] node [below] {$2/9$} (jjv)
		(jjv) edge [bend right=15,looseness=0.5] node [above] {$0/9$} (iiv);
\end{tikzpicture}
\end{center}

The tables of multipliers and preferences are:
\mbox{}\\

\begin{table}[!htb]
    \begin{subtable}{.5\linewidth}
      \centering
        \begin{tabular}{ | l | l | l | l | l | l | l | l | l |}
    \hline
    $i$ & $i_s$ & $i_u$ & $i_u$ & $i_v$ & $i_u$ & $j_u$ & $i_v$ & $j_v$\\ \hline
    $j$ & $j_u$ & $j_v$ & $j_t$ & $j_u$ & $j_u$ & $i_u$ & $j_v$ & $i_v$\\ \hline
    $\mu(ij)$ & $\frac{4}{16}$ & $\frac{45}{8}$ & $\frac{15}{12}$ & $\frac{18}{16}$ & $\frac{15}{16}$ & $\frac{15}{16}$ & $\frac{9}{8}$ & $\frac{9}{8}$\\ \hline
  	\end{tabular}
    \end{subtable}%
    \begin{subtable}{.5\linewidth}
      \centering
        \begin{tabular}{ | l | l |}
    \hline
    $i$/$j$ & preference \\ \hline
    $i_u$ & $i_u j_u \succ_{i_u} i_u j_v \succ_{i_u} i_u j_t \succ_{i_u} j_u i_u$ \\ \hline
    $j_u$ & $j_u i_u \succ_{j_u} i_v j_u \succ_{j_u} i_s j_u \succ_{j_u} i_u j_u$ \\ \hline
    $i_v$ & $i_v j_v \succ_{i_v} i_v j_u \succ_{i_v} j_v i_v$ \\ \hline
    $j_v$ & $j_v i_v \succ_{j_v} i_u j_v \succ_{j_v} i_v j_v$ \\ \hline
  	\end{tabular}
    \end{subtable} 
\end{table}

We can see that not only equation~\ref{eq7}, \ref{eq8}, and \ref{eq9} but also all the conditions in Definition~\ref{def5.2} are satisfied. The numerator of the stable allocation corresponds to the stable flow.

\end{example}

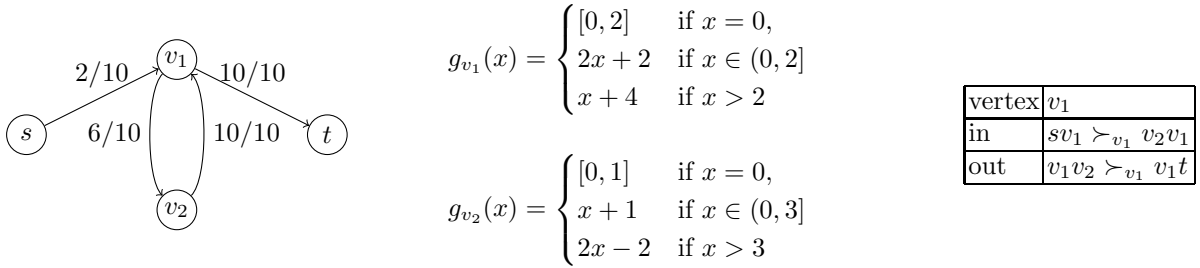
\end{document}